\documentclass[11pt,twoside,a4paper]{article}
\usepackage[margin=1in]{geometry}
\usepackage{multirow,bm}
\usepackage[justification=centering]{caption}
\usepackage[justification=centering]{subcaption}
\usepackage{hyperref}
\usepackage{url}
\usepackage{tikz}
\usetikzlibrary{arrows}
\usepackage[]{algorithm2e}
\usepackage{amsmath,amsthm} 
\usepackage{amssymb} 
\usepackage{mathrsfs}

\newtheorem*{theorem*}{Theorem}
\newtheorem{Theorem}{Theorem}
\newtheorem{Conjecture}{Conjecture}
\newtheorem{Notation}{Notation}
\newtheorem{Proposition}{Proposition}
\newtheorem{Lemma}{Lemma}
\newtheorem{Definition}{Definition}
\newtheorem{Remark}{Remark}
\newtheorem{Corollary}{Corollary}
\newtheorem{Example}{Example}

\newcommand{\F}{\mathbb{F}}

\def\R{\mathbb{R}}

\newcommand{\word}[1]{\ensuremath{\boldsymbol{#1}}}
\newcommand{\uv}{\word{u}}
\newcommand{\xv}{\word{x}}
\newcommand{\vv}{\word{v}}

\newcommand{\Rel}{\mathrm{Rel}}
\newcommand{\Avr}{\mathrm{Avr}}

\newcommand{\Mon}{\mathcal{M}_m}

\newcommand{\eqdef}{\stackrel{\text{def}}{=}}
\newcommand{\indices}{\text{\rm ind}}

\newcommand{\eval}{\mathsf{ev}}
\newcommand{\CC}{\mathscr{C}}

\newcommand{\CS}[1]{\mathcal{C}_{#1}}
\newcommand{\CN}[2]{\ensuremath{\mathcal{C}_{#1,#2}}}

\newcommand{\Ns}{\bm{N}}

\newcommand{\Css}{\bm{C}}

\newcommand{\Yc}{\mathcal{Y}}
\newcommand{\BEC}{\mathrm{BEC}}
\newcommand{\RM}{\mathcal{RM}}

\newcommand{\LAR}{\ensuremath{\leq_{\mathrm{Avr}}}}

\newcommand{\card}[1]{|#1|}

\newcommand{\Bha}[1]{{\mathcal{B}\left( #1 \right)}}
\newcommand{\W}[2]{W^{#1}_{#2}}
\newcommand{\Wb}[2]{\left(W^{#1}_{#2}\right)}
\newcommand{\Wg}[3]{{{#1}^{#2}_{#3}}}
\newcommand{\Wgb}[3]{\left({{#1}^{#2}_{#3}}\right)}
\newcommand{\weako}{\preceq_{\mathrm{w}}}

\title{Bhattacharyya parameter of monomials codes for the Binary Erasure Channel: from pointwise to average reliability}

\author{Vlad-Florin Dr\u{a}goi $^{1,2}$ and Gabriela Cristescu $^{1}$}

\date{\footnotesize
$^{1}$Faculty of Exact Sciences, Aurel Vlaicu University of Arad, Romania; \\\{vlad.dragoi,gabriela.cristescu\}@uav.ro\\
$^{2}$LITIS, University of Rouen Normandie, France.
}

\begin{document}
\maketitle

\begin{abstract}Monomial codes were recently equipped with partial order relations, fact that allowed researchers to discover structural properties and efficient algorithm for constructing polar codes. Here, we refine the existing order relations in the particular case of Binary Erasure Channel. The new order relation takes us closer to the ultimate order relation induced by the pointwise evaluation of the Bhattacharyya parameter of the synthetic channels. The best we can hope for is still a partial order relation. To overcome this issue we appeal to related technique from network theory. Reliability network theory was recently used in the context of polar coding and more generally in connection with decreasing monomial codes. In this article, we investigate how the concept of average reliability is applied for polar codes designed for the binary erasure channel. Instead of minimizing the error probability of the synthetic channels, for a particular value of the erasure parameter $p$, our codes minimize the average error probability of the synthetic channels. By means of basic network theory results we determine a closed formula for the average reliability of a particular synthetic channel, that recently gain the attention of researchers.\\
\end{abstract}

\section{Introduction}
One of the most striking development in coding theory in the last two decades is probably the theory around polar codes. In his seminal article \cite{A09}, Arikan demonstrated, for the first time, that one could achieve the capacity of Binary Discrete Memoryless Channels (BDMC) using both efficient encoding as well as efficient decoding algorithms. The so called polar codes are now present in the fifth generation (5G) technology \cite{Land2020}. Indeed, polar code was elected as the standard coding technique for the control channel in support of the enhanced Mobile BroadBand service, one of the major parts in the 5G wireless network technology. Getting back to the three principal directions on which coding theory evolved, \emph{polar coding} seemed to be unrelated to the classical \emph{algebraic coding}. Typically, the construction of polar codes does not
come from any particular structure in the code but rather from the process of channel polarization. However, polar codes are closely related to Reed-Muller codes, as pointed out even by Arikan \cite{A09}. Hence, polar and Reed-Muller code share a common algebraic description \cite{D17, BDOT16}. More precisely, they are sub-classes of a bigger family of algebraic codes called \emph{decreasing monomial codes} (DMC). The structure underlying DMCs and its algebraic formalism were applied in conjunction with other fields, e.g., in the context of quantum error correcting codes \cite{rengaswamy2020,rengaswamy2020b,krishna2019}, post-quantum cryptography \cite{BCDOT16,DBC19,BDK17}, network reliability \cite{DB2019,DCB18,BCD21}.

Several challenges regarding polar coding, among which efficient construction of polar codes given a specific BDMC,  were proposed. Arikan's initial technique \cite{A09} was improved by several authors \cite{M_2017, He_2017, TV13, MT09,  MELK13, KSU10, AD15, Trifonov2017, Huang2020}. Let $W$ denote a BDMC, $m$ a fixed integer and $\uv$ a binary vector of length $m.$ The main idea in the construction of polar codes is to estimate the reliability of the synthetic channels $\{W^{\uv}\;|\;\uv\in\{0,1\}^m\}$. For that one might use the Bhattacharyya parameter $\Bha{W^{\uv}}(p)$, where $p$ denotes the error probability of the channel $W$. The message
bits of a polar code of length $2^m$ and dimension $k$ are allocated to the $k$ sub-channels $W^{\uv}$ having the smallest $\Bha{W^{\uv}}$. Hence, one might classify the set of $W^{\uv}$ into "good" (reliable) or "bad" (non-reliable). For a fixed value of $p$ the values $\Bha{W^{\uv}}(p)$ are totally ordered. In other words, when the parameter $p$ is fixed any distinct pair of channels $W^{\uv},W^{\vv}$ satisfy either $\Bha{W^{\uv}}(p)\leq \Bha{W^{\vv}}(p)$ or $\Bha{W^{\vv}}(p)\leq \Bha{W^{\uv}}(p).$ In this case, we say that a channel $W^{\uv}$ is point-wise more reliable than a channel $W^{\vv}.$ However, when considering the whole interval $p\in[0,1]$, ranking the synthetic channels becomes complicated. In this case we say that $W^{\uv}$ is globally more reliable than $W^{\vv}$, and write $\uv\leq\vv$, if and only if $\forall p\in[0,1]\;,\;\Bha{W^{\uv}}(p)\leq \Bha{W^{\vv}}(p).$ 

One of the most efficient techniques that orders the set of synthetic channels (with respect to the concept of globally more reliable), provides sub-linear complexity construction \cite{M_2017}. It exploits the existence of a partial order (denoted by $\preceq$) on the set of synthetic channels \cite{BDOT16}. This partial order is compatible with the notion of globally more reliable, i.e., $\uv\preceq\vv\Rightarrow \uv\leq \vv.$ Relation $\preceq$ was also applied in other contexts. In \cite{BDOT16} the authors proved fundamental properties of polar codes by means of an algebraic formalism. The order $\preceq$ played a crucial role in proving that code-based cryptosystems employing polar codes are insecure \cite{BCDOT16, DBC19}. Even though $\preceq$ provided a contribution to understanding polar codes, i.e., their structure and construction, simulations show that $\preceq$ is far from ordering $\Bha{W^{\uv}}$ optimally. Hence, in a recent article $\preceq$ was refined \cite{WFS17}. 

In the analysis of the performance of several families of codes, among which polar, Reed-Muller, cyclic and BCH codes, the communication channel that received a lot of attention is the Binary Erasure Channel ($\BEC$). When polar codes are designed for $\BEC(p)$ (in this particular case $p$ denotes the erasure probability) all the synthetic channels $\{W^{\uv}\;| \; \uv\in\{0,1\}^m\}$ are also $\BEC$. In this case, the erasure probability of $W^{\uv}$ is equal to the Bhattacharyya parameter of $W^{\uv}.$ Here, we analyze this particular channel. Our choice is motivated by several results and methods. First of all the simplicity of this channel makes the theoretical proofs significantly simpler and easier. Also, many of the properties that hold for the $\BEC$ turn out to be valid for more general channel models. For example, the proof of Reed-Muller codes achieving the capacity of a communication channel started with the $\BEC$ \cite{KMSU17,SSV17}. Codes that admit a doubly-transitive automorphism group or having large orbits under the action of their permutation group achieve the capacity of the $\BEC$ \cite{KMSU17,KCP2016}. In \cite{Roth2019} the authors analyze threshold points for $W^{\uv}$ in the case of $\BEC$, fact that allows them to propose sets of asymptotically "good" channels. Recently, in \cite{BD2020} the authors analyzed the Bhattacharyya parameter of polar codes for the $\BEC$ using network reliability theory. They have proposed simple approximations of $\Bha{W^{\uv}}$. These were used to determine sub-intervals of $[0,1]$ where polar codes coincide with Reed-Muller codes. They have also managed to determine new sets of asymptotically "good" channel.

\subsection{Polar codes are strongly decreasing monomial codes}
Polar codes over the $\BEC$ satisfy an order relation that is finer than $\preceq.$ Hence, we define another order relation $\preceq_d$ on the set of monomials on $m$ variables\\ $\Mon=\{\bm{1},x_0,\dots,x_{m-1},x_0x_1,\dots, x_0x_1\dots x_{m-1}\}$, coming closer to the $\leq$ relation, i.e., we have
\begin{equation*}
    \forall f,g\in\Mon\quad f\preceq g\Rightarrow f\preceq_d g\Rightarrow f\leq g.
\end{equation*}
The relation $\preceq_d$ allows to compare monomials with equal degrees that were not comparable with respect to $\preceq$, e.g., $x_1x_2\preceq_d x_0x_3.$ The idea of $\preceq_d$ came from the link between the set of monomials of degree $d$ in $\Mon$ and the set of partitions/Young diagrams inside the $d\times(m-d)$ grid (see Proposition 3.7.8 in \cite{D17}). From that, looking at order relations on partitions came as a natural idea, and the most common one is the dominance order \cite{S_2011}. The order $\preceq_d$ is exactly defined as the dominance order on partitions inside a fixed grid.     

The main result in this section can be stated as follows

\begin{theorem*}
Polar codes over the Binary Erasure Channel are strongly decreasing monomial codes.
\end{theorem*}
In the proof of this theorem we will need to demonstrate two useful properties of this new order
\begin{itemize}
    \item given two monomials in $f,g\in\Mon$ such that $f\preceq_d g$ then for any multiples $fh,gh$ with $\gcd(h,f)=1$ and $\gcd(g,h)=1$ we have $fh\preceq_d gh,$ where $\gcd(f,g)$ denotes the greatest common divisor of $f,g$.   
    \item two particular monomials are the key ingredients in the proof, $x_1x_2\preceq_d x_0x_3.$ We show that for all $p\in[0,1], \Bha{\W{x_1x_2}{}}(p)\leq \Bha{\W{x_0x_3}{}}(p)$, and in general that any pair of monomials of degree 2 $f,g$,  satisfying $f\preceq_d g$ has the property for all $p\in[0,1], \Bha {W^{f}}(p)\le\Bha{W^{g}}(p).$  
\end{itemize}
Even though $\preceq_d$ get us closer to the ultimate order relation $\leq$ we know that $\preceq_d$ is a partial order relation. $\preceq_d$ seems to perform as well as the order relation from \cite{WFS17}, being much simpler to describe and analyze that the order in \cite{WFS17}. Also, in \cite{WFS17} the authors determine new order relations based on some hypothesis which are not algebraically easy to express, and which are to be tested each time we change the parameters of the code.    

\subsection{Average reliability of synthetic channels} 
Hence, we are still left with elements that are not comparable and for which we need to compute $\Bha{W^{\uv}}.$ In order to overcome this issue we propose an alternative solution. Suppose that the erasure probability of the channel $p$ changes with respect to the uniform distribution over the closed interval $[0,1].$ Instead of constructing, for each $p$, the corresponding polar code, we propose to construct the best polar code in average. More exactly, we consider the average reliability of the synthetic channels $W^{\uv}$, $\Avr(W^{\uv})=\int_{0}^{1}\Bha{W^{\uv}}dp$, and choose those $\uv$ that minimize this quantity. As the average reliability induces a total order relation (see Figure \ref{fig:Avr}) there is only one polar code for a given dimension and length. It is the linear code that minimizes the average error probability for all $p\in[0,1].$ Hence, it might be less efficient than polar codes designed for a particular value of $p$ but it has the best performance in average. 
\begin{figure}[!h]
    \centering
    \begin{subfigure}{.32\textwidth}
  \includegraphics[width=\textwidth]{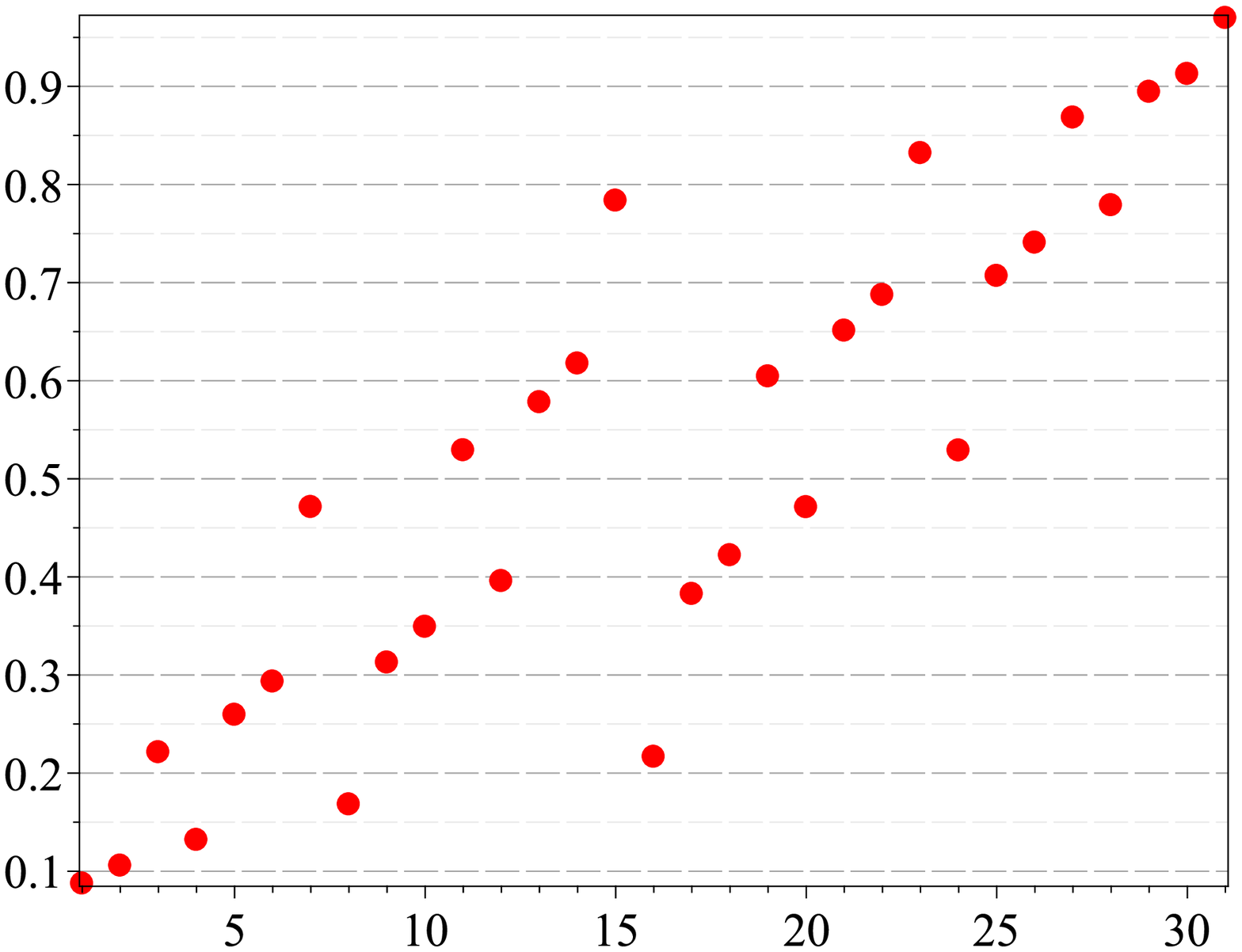}  
  \caption{$m=5$}
  \label{fig:avr5}
\end{subfigure}
    \begin{subfigure}{.32\textwidth}
    \includegraphics[width=\textwidth]{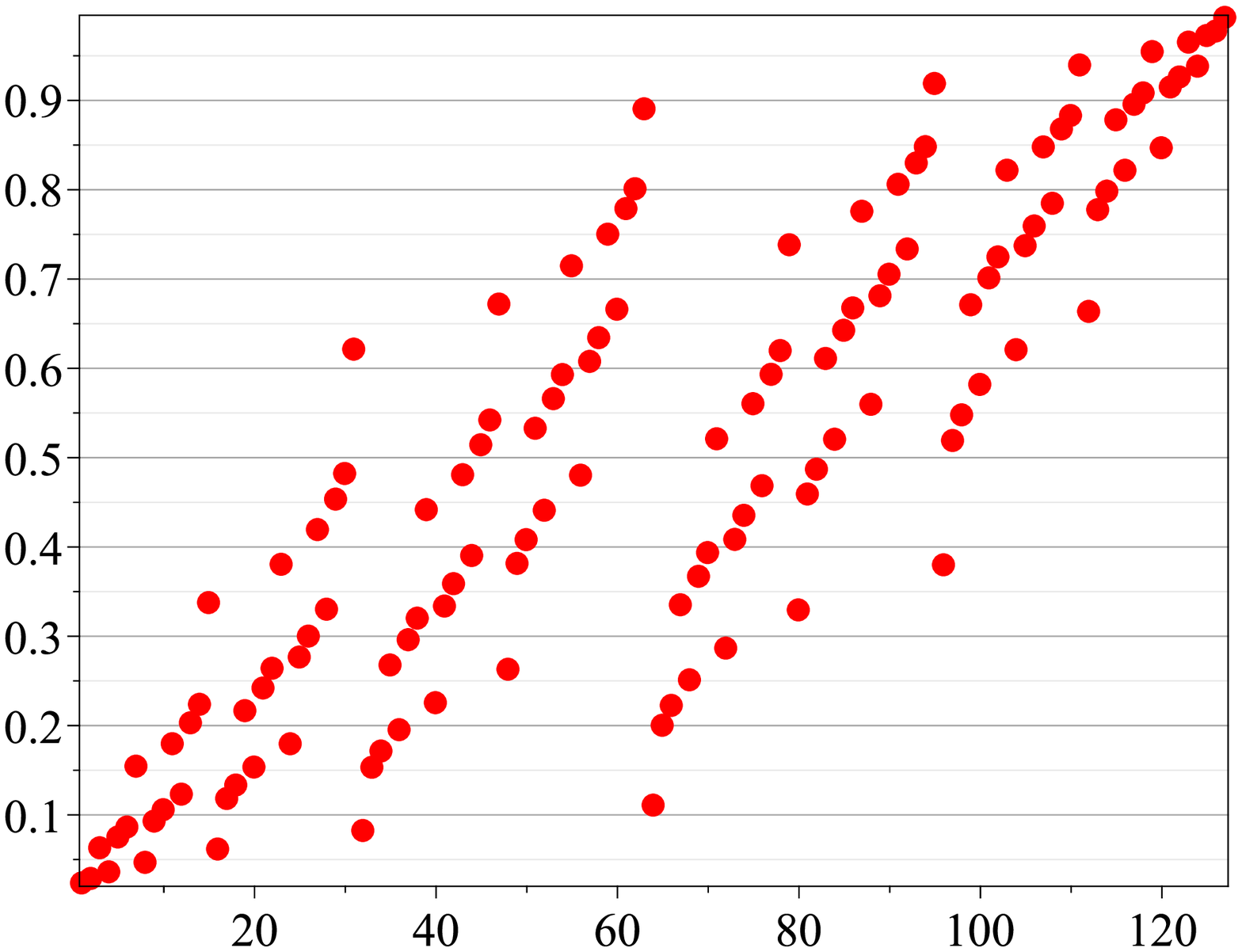}
    \caption{$m=7$}
    \label{fig:avr7}
    \end{subfigure}
\begin{subfigure}{.32\textwidth}
    \centering
    \includegraphics[width=\textwidth]{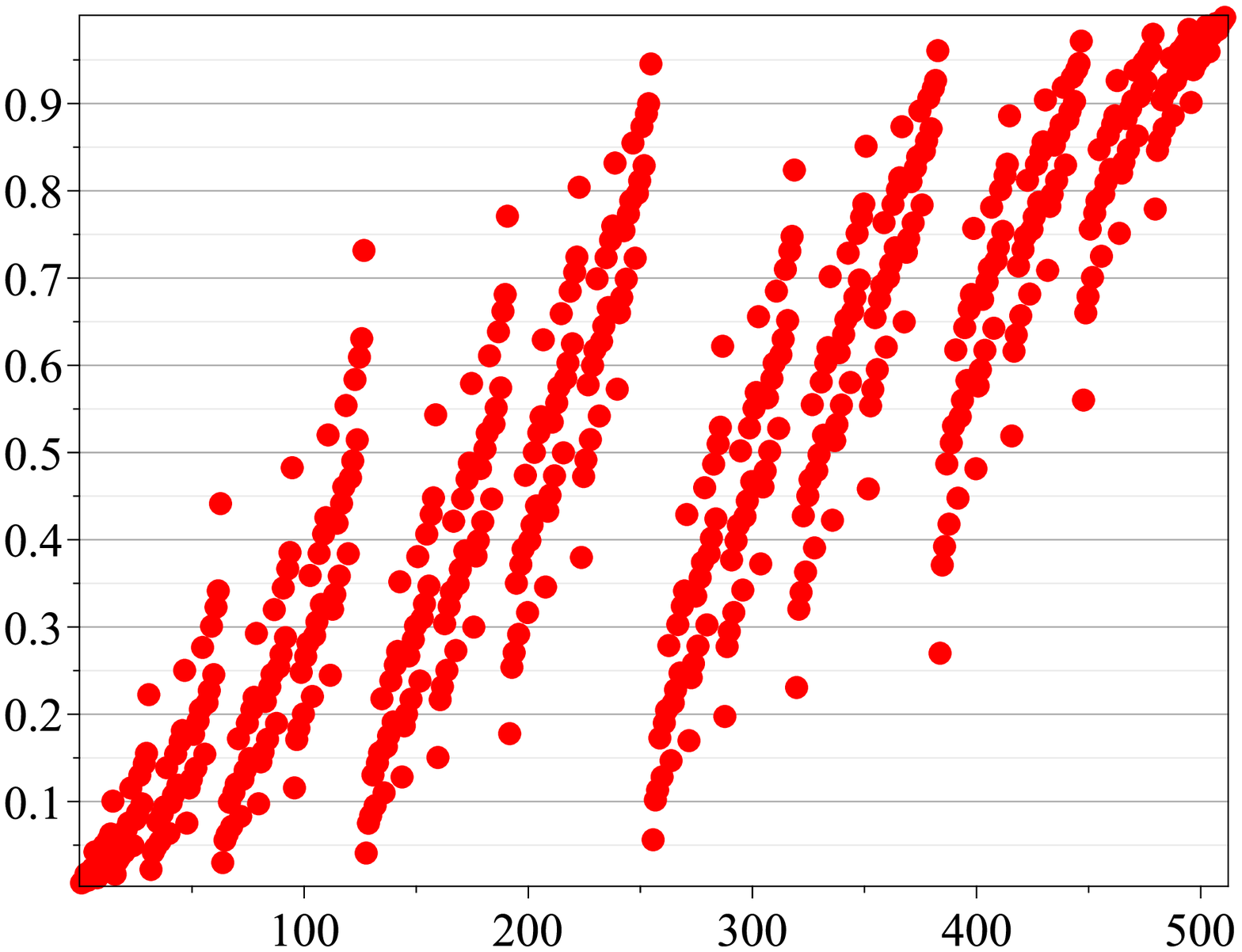}
    \caption{$m=9$}
    \label{fig:avr9}
    \end{subfigure}
    \caption{Average Bhattacharyya parameter. On the x-axis are the integer values of the binary vectors $\uv\in\{0,1\}^m$, and on the y-axis are the values $\Avr(\Bha{W^{\uv}}).$  }\label{fig:Avr}
\end{figure}

The preorder $\uv\LAR\vv\Leftrightarrow \Avr({W^{\uv}})\leq\Avr({W^{\vv}})$  induces a complementarity property with respect to the integral operator over $[0,1]$ as defined in \cite{CD2020,CD2021} in case of two-terminal networks. We retrieve a similar property, i.e., $\Avr({W^{\uv}})=1-\Avr\left(W^{\overline{\uv}}\right),$ where $\overline{\uv}$ is the bit-wise complement of $\uv$, in the context of monomial codes. Our simulations have shown that, considering the relation $\LAR$  in the set of the synthetic channels, in each sub-interval $(i/10,(i+1)/10)$, for $0\leq i\leq 9$, we have a rough proportion of $2^m/10$ binary vectors $\uv.$ So, roughly speaking an uniform distribution could be used to approximate the number of $\uv$ inside each sub-interval, with respect to $\Avr.$ However, our result is not constructive, in the sense that it does not fully characterize exactly the $\uv$ that belong to a specific interval. An answer to this question might provide an extremely efficient method for constructing polar codes and give much more insight on the synthetic channels $W^{\uv}.$  

\subsubsection{Threshold points for sharp transitions} 
Determining the threshold point of $\Bha{W^{\uv}}$ is in general a difficult task \cite{mondelli2016,Roth2019}. In \cite{Roth2019} the authors analyze a particular synthetic channel $W^{(1^i0^{m-1})}$, for which asymptotic threshold points were determined. The conditions on $i$ and $m$ were further improved in \cite{WFS17}. Based on some basic notions and facts from network theory we determine an exact formula for the average reliability of $W^{(1^i0^{m-1})}$. The main result is
\begin{theorem*}Let $\uv=(1^i0^{m-i}).$ Then
\begin{equation}
\Avr\left(W^{\uv}\right)=1-\dfrac{1}{\binom{2^i+2^{i-m}}{2^i}}
\end{equation}
\end{theorem*}

This allows us to determine the exact threshold point of this particular channel. Moreover, we demonstrate that for any $i\leq m-\log_2(m)-\log_2(\log_2(m))$ the channel $W^{(1^i0^{m-i})}$ has an average Bhattacharyya parameter that tends to zero when $m$ goes to infinity, i.e., $W^{(1^i0^{m-i})}$ is asymptotically "good" in average. Another consequence of our formula is that for any monomial $g\preceq_d x_{m-i+1}\dots x_{m}$ with $i\leq \log_2(\log_2(m))$ is such that $\Avr(W^{g})$ tends to zero when $m$ goes to infinity.    

Another significant implication of our result is that any synthetic channel in the $\RM(i,m)$ is asymptotically "good" in average, for any $i\leq \log_2(\log_2(m)).$ 
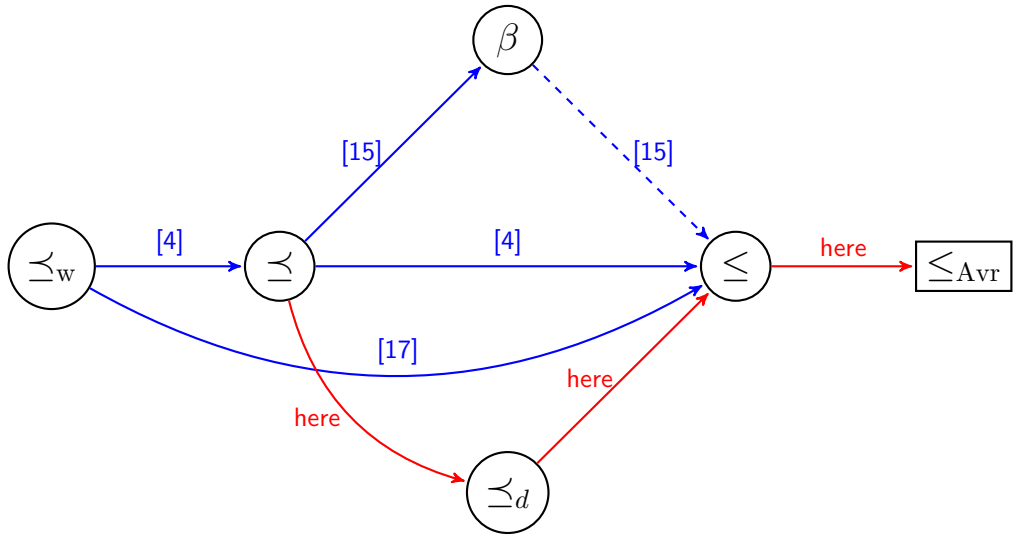
\begin{figure}[!h]
    \centering
\begin{tikzpicture}[->,>=stealth',shorten >=1pt,auto,node distance=3cm,
                    thick,main node/.style={circle,draw,font=\sffamily\Large\bfseries}]

  \node[main node] (1) {$\weako$};
  \node[main node] (2) [right of=1] {$\preceq$};
  \node (3) [right of=2] {};
    \node[main node] (7) [below of=3] {$\preceq_d$};
  \node[main node] (4) [right of=3] {$\leq$};
  \node[main node, rectangle] (5) [right of=4] {$\LAR$};
  \node[main node] (6) [above of=3] {$\beta$};

  \path[every node/.style={font=\sffamily\small}]
    (1) edge[blue] node [above] {\cite{BDOT16}} (2)
		edge[blue, bend right] node[above] {\cite{MT09}}(4)
    (2) edge[blue] node [above] {\cite{BDOT16}} (4)
        edge[blue] node [left]{\cite{He_2017}}(6)
    (6) edge[dashed,blue] node [right]{\cite{He_2017}} (4)
    (2) edge [red, bend right] node [left] {here} (7)
    (7) edge [red] node [left] {here} (4)
    (4) edge [red] node [above] {here} (5);
    \end{tikzpicture}
    \caption{Order ($\weako,\preceq,\preceq_d,\leq$) and pre-order relations ($\LAR$) for monomials codes over the $\BEC$. The connections in red are the results coming from this article. The dotted edge from $\beta$ to $\leq$ represents an order relation that is valid only for a sub-interval of $[0,1].$}
    \label{fig:order-rel-BEC}
\end{figure}
\clearpage
  
\section{Background and preliminary results}

Let us begin by listing some of the usual notations from coding theory that are going to be used in this article. $\F_2$ will denote the finite field with two elements $\{0,1\}.$ 
Let $k, n$ be two strictly positive integers and $k\leq n.$ A code $\CC$ of length $n$ and dimension $k$ is a vector sub-space of $\F_2^n$ of dimension $k.$ In this article, we focus out attention on a particular family or linear codes, namely monomial codes. $W$ will be used to denote a communication channel with binary input $x\in \F_2$ and output from an alphabet $y\in\Yc.$ In particular, we will focus on $\BEC(p)$, where the output is $\Yc=\{0,1,?\}$, $?$ denoting an erasure and $p$ being the erasure probability. For a more detailed reading of the subject we recommend \cite{RU08,R06}. 

\subsection{Monomial codes}
Monomial codes are a special class of structured codes. Informally, any code that admits a basis, in which each vector is the evaluation of a monomial, is called a monomial code. In general, monomial codes have a predefined length, i.e., $2^m.$ Many of the notations, definitions, properties, and results presented in this section are taken from \cite{D17}.  

In this article, binary vectors of length $m$ will be denoted using bold small letters, e.g., $\uv=(u_{0},\dots,u_{m-1})\in\F_2^m$, with the convention that bits are ordered from left to right, $u_0$ being the least significant bit. We also define the bit-wise complement of $\uv\in \{0,1\}$ by $\overline{\uv}=\bm{1_m}\oplus\uv$ (as in \cite{BDOT16}), where $\bm{1_m}$ is the all-ones vector. 
 The set $\{\uv \in \F_2^m\}$ will be ordered in a natural manner, using the mapping 
\begin{equation*}(u_0,\dots,u_{m-1})\rightarrow u=\sum_{i=0}^{m-1}u_i 2^i,
\end{equation*} and the natural order on the integers. Notice that we compute the value $u$ regardless of the fact that $u_i\in \F_2.$ Notice that the relation between $\uv$ and $\overline{\uv}$ induces $u+\overline{u}=2^m-1$. 

We consider multivariate polynomials and monomials defined over the polynomial ring $\R_m={\F_2[x_0,x_1,\dots,x_{m-1}]}/{(x_0^2-x_0,\dots,x_{m-1}^2-x_{m-1})}.$ The usual operators will be employed, i.e., for $f,g\in\R_m$, we denote by $\deg{f}$ the degree of $f$, $\gcd(f,g)$ the greatest common divisor of $f$ and $g$. $f/g$ denotes the quotient of $f$ and $g$.  
\begin{Notation}Let $m$ be a strictly positive integer. We denote 
\begin{itemize}
\item monomials: $\xv^{\uv}=x_{0}^{u_{0}}\cdots{} x_{m-1}^{u_{m-1}},$ where $\uv\in\F_2^m.$ 
\item support of a monomial: $\indices(g)=\{l_1\,\dots,l_s\}$ , where $g=x_{l_1}\dots x_{l_s}$ and $0\le l_1<l_2\dots<l_s\le m-1.$
\item a subset of the support of a monomial: $g_{[0,s]}=\gcd(g,\prod_{i=0}^s x_i).$
\item the set of monomials: $\Mon \eqdef \left \{\xv^{\uv}~|~ \uv=(u_{0},\dots{},u_{m-1})\in\F_2^m \right \}.$
\end{itemize}
\end{Notation}

\begin{Proposition}[\cite{C10}]
Let $g\in \R_m$ and order the elements in $\F_2^m$ with respect to the decreasing index order. Define the evaluation function
\[
\begin{array}[h]{ccccc}
\R_m    & \to &\F_2^{2^m}\\
g& \mapsto &\eval(g) = \big(g(\uv) \big)_{\uv \in \F_2^m}
\end{array}
\]
Then $\eval$ is a bijection defining an isomorphism between the vector spaces $(\R_m,+,\cdot)$ and $(\F_2^n,+,\cdot).$
\end{Proposition}

Now, we are ready to define the concept of monomial codes. 
\begin{Definition}[Monomial code]
Let $I\subseteq\Mon$ be a finite set of monomials in $m$ variables. 
The linear code  defined by $I$ is the vector subspace $\CC(I) \subseteq \F_2^{2^m}$ 
generated by $\{ \eval(f) ~|~ f \in I\}$ that is called \emph{monomial code}.
\end{Definition}

\begin{Proposition}[\cite{D17}]\label{prop:dim_monomial}
For all $I \subseteq \Mon$ the dimension of the monomial code $\CC(I)$ is equal to $\card{I}$.
\end{Proposition}

\begin{Remark}
 The $r^{th}$ order Reed-Muller code $\RM(r,m) \eqdef \big \{\eval(g) ~|~ g \in\R_m, \deg g \leq r \big\}$ is a monomial code with dimension $k=\sum_{i=0}^r \binom{m}{i}.$
\end{Remark}


\subsection{Polar codes}

In order to define polar codes we have to introduce the concept of synthetic channels. Consider the channel transformation $W \rightarrow (W_2^{(0)} , W_2^{(1)} )$ defined in the following manner. 

 \begin{Definition}[Synthetic channels]\label{def:W_moins_plus}
Let $W$ be a BDMC with output alphabet $\Yc$and $x_1,x_2\in\F_2$ be the inputs and $y_1,y_2\in\Yc$ be the outputs of two copies of $W.$ Define two new channels
 \begin{eqnarray*}
  {W^{(1)}}(y_1,y_2|x_2) & \stackrel{def}{=} & \frac{1}{2} \sum_{x_1 \in \F_2} W(y_1|x_1) W(y_2|x_1 \oplus x_2)\\
 {W^{(0)}}(y_1,y_2,x_2|x_1) & \stackrel{def}{=} & \frac{1}{2} W(y_1|x_1) W(y_2|x_1 \oplus x_2). 
 \end{eqnarray*}
 \end{Definition}

 For any $\uv=(u_0,\dots,u_{m-1})\in \{0,1\}^m$ we define $W^{\uv}=((W^{u_{m-1}})^{\dots})^{u_0}$ as in \cite{BDOT16}. Also, we extend the notation to monomials, by $\W{f}{m}=W^{\uv}$ where $f=\xv^{\uv}\in\Mon.$ We are using the index $m$ in $\W{f}{m}$ to precisely identify the number of variables on which $f$ is expressed. For example, if $\uv=(1,0,0,1,1)$ we have $\W{f}{5}=\W{x_0x_3x_4}{5}=\Wg{\Wb{x_4}{1}}{x_0x_3}{4}=\Wg{\Wgb{\Wb{x_4}{1}}{x_3}{1}}{x_0}{3}.$ 
 
 \begin{Definition}\label{def_Bha}
Let $W$ be a BDMC with output alphabet $\Yc.$ Then the Bhattacharyya parameter of the channel $W$ is 
 \begin{equation}
 \Bha{W} = \sum_{y \in \Yc} \sqrt{W(y|0) W(y|1)}.
 \end{equation}
 \end{Definition}
 \begin{Remark}
 Let $W$ be a $\BEC(p)$ then we have that $\Bha{W^{x_0}}=\Bha{W^{(1)}}=2p-p^2$ and $\Bha{W^{\bm{1}}}=\Bha{W^{(0)}}=p^2.$ 
 \end{Remark}

 \begin{Definition}\label{def:polar_code}
 The polar code of length $n=2^m$ and dimension $k$ devised for the channel $W$ is the linear code obtained by selecting the set of $k$ synthetic channels with the smallest $\Bha{W^{\uv}}$ values among all $\uv\in\{0,1\}^m$.
 
 Moreover, we define the relation \begin{equation}f\leq g \Leftrightarrow \uv\leq \vv \Leftrightarrow \Bha{W^{{\uv}}}(p)\leq \Bha{W^{{\vv}}}(p), \forall p\in [0,1].\label{eq:rel-order}\end{equation} 

 \end{Definition}
The relation \eqref{eq:rel-order} is called \emph{universal}, i.e., two monomials $f,g$ satisfying $f\leq g$ are always comparable for any $p$ and any $m.$ This property can be used when constructing polar codes by storing a table with all such monomials. However, there might be several monomials bigger than $f$ which are not comparable pairwise (see for example \cite{BDOT16,M_2017}). Indeed, one can easily verify that $\leq$ is a well-defined order relation (reflexive, anti-symmetric, and transitive), and thus, induces a poset on the set of monomials. In some particular cases, the order $\leq$ becomes total (all elements are ordered in a chain), e.g., when $W=\BEC$ and $m\leq 4$. However, in general $\leq$ is a partial order, even in the case of $W=\BEC$ (starting from $m=5$), as pointed out in \cite{DB2019,BD2020,D17}. 
\begin{Proposition}
Let $m\ge 5$ be an integer and $W=\BEC.$ Then $\{\Mon,\leq\}$ is a poset. 
\end{Proposition}

For simplification, when we refer to ordering the Bhattacharyya parameters we will just write $\Bha{W^{{\uv}}}\leq \Bha{W^{{\vv}}}.$  
\subsection{Weakly decreasing and decreasing monomial codes}
\begin{Definition}\label{def:order}
Let $f$ and $g$ be two monomials in $\Mon.$ 
\begin{itemize} 
\item The $\weako$ order between $f$ and $g$ is defined as  
\[f \weako g\quad \text{iff} \quad f|g.\]
\item The $\preceq$ order between $f$ and $g$ is defined as
				\begin{itemize}
				\item when $\deg(f)=\deg(g)=s$ and $f=x_{i_1}\dots x_{i_s}$, $g=x_{j_1}\dots x_{j_s}$ we have \[f\preceq g\quad \text{iff}\quad \forall\;1\leq\ell\leq s\quad  i_\ell \le j_\ell.\]
				\item when $\deg(f)<\deg(g)$ we have \[f\preceq g\quad \text{iff}\quad \exists g^*\in \Mon\;\text{s.t.}\; f\preceq g^*\weako g.\]
				\end{itemize}
\end{itemize}			
\end{Definition}

The two order relations $\weako$ and $\preceq$ are well defined. $\weako$ was already used in the case of polar codes but in a completely different context by Mori and Tanaka in \cite{MT09}. In their case the purpose was to tighten the bounds of the error block probability of a polar code designed for the $\BEC$ family.

 Notice that $\weako$ is weaker than $\preceq$, meaning that $\forall f,g\in\Mon f\weako g\Rightarrow f\preceq g.$ The converse is not always true: taking, for example, $f=x_0x_2$ and $g=x_1x_2$ it follows by definition that $f\preceq g$ but $f\not\weako g.$ We also remark that $\bm{1}$ is the smallest element both for $\preceq$ and for $\weako$, and we have
\[\bm{1}\preceq x_0\preceq x_1\preceq \dots \preceq x_{m-1}.\]

\begin{Definition}\label{def:weak_set}
Let $f$ and $g$ be two monomials in $\Mon$ such that $f\preceq g$ and $I\subset \Mon.$
\begin{itemize}
\item We define the closed interval $[f,g]_{\preceq}=\{h\in\Mon\;|\;f\preceq h\preceq g\}.$
\item $I\in\Mon$ is called a decreasing set if and only if ($f \in I$ and $g \preceq f $) implies $g \in I$.
\item Let $I\in\Mon$ be a decreasing set. Then $\CC(I)$ is called decreasing monomial code.
\end{itemize}
\end{Definition}

Polar codes were recently related to network theory. In \cite{BD2020} the authors make a connection between the Bhattacharyya parameter of a synthetic channel and the reliability polynomial of a two-terminal network. Following on the same path, we introduce in the next subsection all the required preliminaries in reliability and network theory.

\subsection{Two-terminal networks}

\begin{Definition}
Let $n$ be a strictly positive integer. We say that $\Ns$ is a two-terminal network (2TN) of size $n$ if $\Ns$ is a network made of $n$ identical devices, that has two distinct terminals: an input $S$, and an output $T$. 
\end{Definition}

To any $\Ns$ we associate three parameters: {\it width} ($w$), {\it length} ($l$), and {\it size} ($n$), where $w$ is the size of a ``minimal cut" separating $S$ from $T$, and $l$ is the size of a ``minimal path" from $S$ to $T$. The size of $\Ns$ satisfies
\begin{equation}
n\ge wl \label{eq:nwl}
\end{equation} (see Theorem 3 in \cite{1956_MS_1}). 
When $n=wl$ we say that $\Ns$ is a minimal 2TN \cite{1956_MS_1}.

The composition of $\Ns_1$ and $\Ns_2$ can be defined as in \cite{1956_MS_1}. The resulting network is obtained by replacing each device in $\Ns_1$ by a copy of $\Ns_2$. 
We will denote a composition by $\Css$, the simplest possible being two devices in series $\Css^{(0)}$, and two devices in parallel $\Css^{(1)}$. The composition of $\Css^{(0)}$ with $\Css^{(1)}$ is $\Css^{\uv}=\Css^{(0)}\bullet\Css^{(1)}$,  where $\uv=(0,1)$. The set of all $2^m$-size compositions will be denoted by $\CS{2^m}$, and the set of all compositions of width $2^i$ and length $2^{m-i}$ by $\CN{2^i}{2^{m-i}}$ (see Figure \ref{fig:Net-circ}).

\begin{Proposition}[\cite{2018_DCHGB_J}]\label{pr:param_comp}
Let $m>0$ and $\Css^{\uv}\in \CS{2^m}.$ Then $\Css^{\uv}$ is a minimal 2TN of size $2^m,$ length $l=2^{m-|\uv|}$ and width $w=2^{|\uv|}.$ We also have
$\CS{2^m}=\bigcup_{i=0}^{m}\CN{2^i}{2^{m-i}}.$
\end{Proposition}

\begin{Theorem}[\cite{DB2019}]There is a natural bijection between $\CS{2^m}$ and the set of all $W^{\uv}$, for any fixed positive integer $m.$  
\end{Theorem}
In figure \ref{fig:Circ-Net-3} we illustrate the bijection between the two aforementioned sets. More significant is the equality between the reliability polynomial of a composition $\Css^{\uv}$ and the Bhattacharyya parameter of $W^{\uv}$, fact that is visible from figure \ref{fig:Net-circ} and proved in the next paragraph.

\paragraph{\textit{Reliability polynomial}} The reliability of $\Ns$ is defined as the probability that $S$ and $T$ are connected (also known as $s,t$-connectivity) \cite{C87}. One of the most common hypothesis considered in network theory is that devices are uniformly and identically suppose to close with a probability $p\in[0,1].$ Hence, the reliability of $\Ns$, denoted by $\Rel(\Ns;p)$, can be expressed as a polynomial
\begin{align}\Rel(\Ns;p)&=\sum\nolimits_{i=0}^nN_i(\Ns)\;p^i(1-p)^{n-i}.\label{rel:N_form}
\end{align}

\begin{figure}[!h]
    \begin{subfigure}{.52\textwidth}
 \centering
  \includegraphics[width=\textwidth]{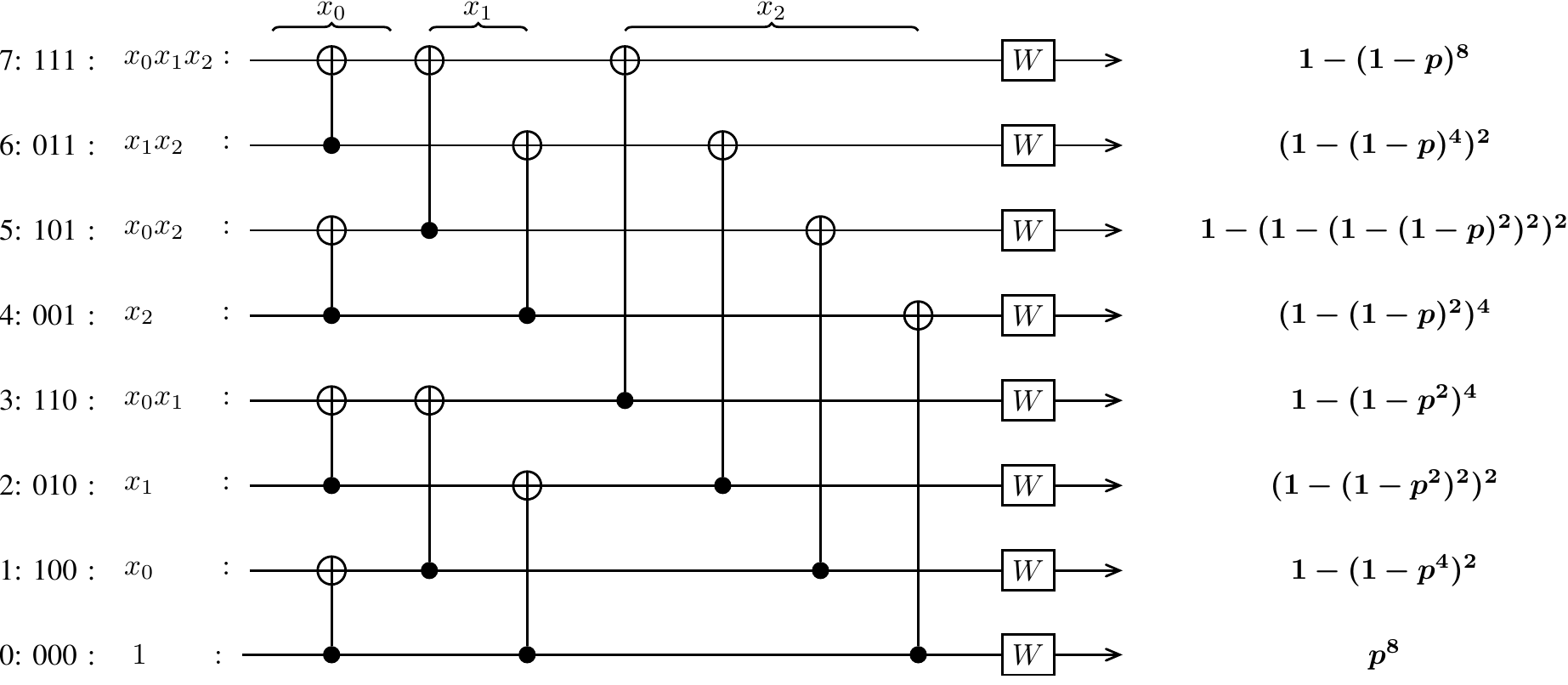}  
  \caption{$W^{\uv}$ and $\Bha{W^{\uv}}$}
  \label{fig:Bha-circuit}
\end{subfigure}
    \begin{subfigure}{.53\textwidth}
\centering
    \includegraphics[width=\textwidth]{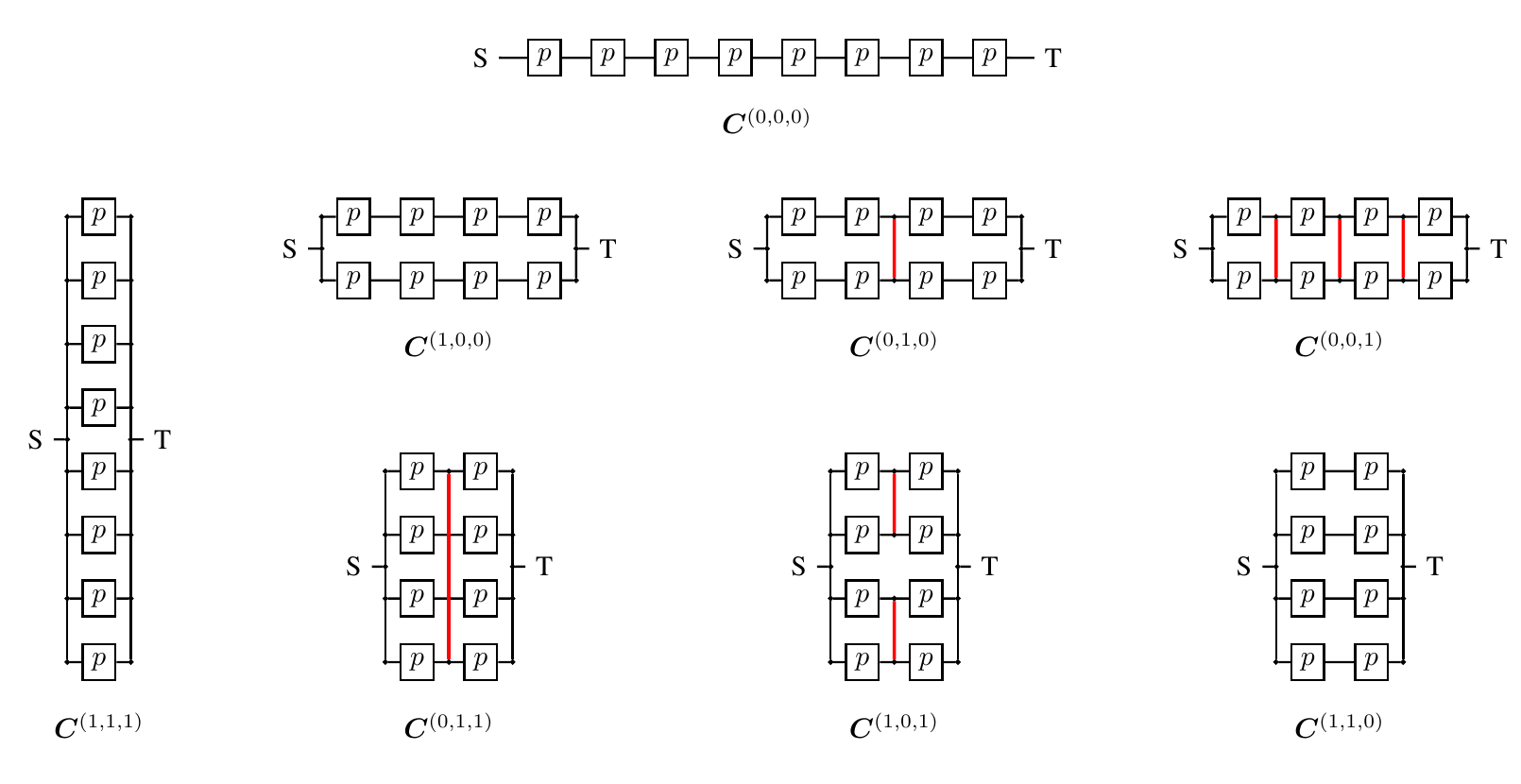}
    \caption{$\CS{3}$}
    \label{fig:Net-circ}
    \end{subfigure}
\caption{Combined circuit as defined by Arikan \cite{A09}, the Bhattacharyya parameter of the corresponding synthetic channels and the compositions in $\CS{3}.$}    \label{fig:Circ-Net-3}
\end{figure}

The coefficients $N_i(\Ns)$ represent the number of paths from $S$ to $T$ of length $i$. Several properties regarding the coefficients $N_i(\Ns)$, as well as complementarity relations between a 2TN $\Ns$ and its dual $\Ns$, are detailed in \cite{CD2020}. 



\subsection{Bhattacharyya parameters and reliability polynomials}
\begin{Theorem}[\cite{DB2019}]\label{thm:reliab_comp}
Let $m>0$ , $\uv\in \{0,1\}^m$, and $W=\BEC(p)$
\begin{equation}
\Bha{W^{{\uv}}}(p)=\Rel(\Css^{\uv};p)
\end{equation}
where $\Rel(\Css^{(0)};p)=p^2$ and $\Rel(\Css^{(1)};p)=1-(1-p)^2$.
\end{Theorem}

 \begin{Proposition}[\cite{WFS17}]Let $m>0$ and $\uv\in \{0,1\}^m.$ Then
\begin{equation}
\Bha{W^{\overline{\uv}}}(p)=1-\Bha{W^{{\uv}}}(1-p).\label{eq:dual}
\end{equation}
\end{Proposition}

This condition expresses the duality of the two corresponding networks, namely $\Css^{\overline{\uv}}$ is the dual of $\Css^{{\uv}}$ (see \cite{CD2020,DB2019}). Notice that by \eqref{eq:dual} one has to analyze only $\uv$ with $|\uv| \leq m/2$.


\section{Polar codes are strongly decreasing monomial code over the BEC}

\begin{Definition}
The $\preceq_d$ order between $f$ and $g$ is defined as
				\begin{itemize}
				\item when $\deg(f)=\deg(g)=s$ and $f=x_{i_1}\dots x_{i_s}$, $g=x_{j_1}\dots x_{j_s}$ we have \[f\preceq_d g\quad \text{iff}\quad \forall\ell \in \{1,\dots,s\} \text{we have}  \sum\limits_{k=0}^{\ell}i_{s-k} \le\sum\limits_{k=0}^{\ell} j_{s-k}
.\]
				\item when $\deg(f)<\deg(g)$ we have \[f\preceq_d g\quad \text{iff}\quad \exists g^*\in \Mon\;\text{s.t.}\; f\preceq_d g^*\weako g.\]
\end{itemize}
\end{Definition}

\begin{Definition}\label{def:strong_set}
Let $f$ and $g$ be two monomials in $\Mon$ such that $f\preceq g$ and $I\subset \Mon.$
\begin{itemize}
\item We define the closed interval $[f,g]_{\preceq_d}=\{h\in\Mon\;|\;f\preceq_d h\preceq_d g\}.$
\item $I\in\Mon$ is called a strongly decreasing set if and only if ($f \in I$ and $g \preceq_d f $) implies $g \in I$.
\item Let $I\in\Mon$ be a strongly decreasing set. Then $\CC(I)$ is called strongly decreasing monomial code.
\end{itemize}
\end{Definition}

\begin{Lemma}
The order $\preceq_d$ is a well-defined order relation and $\{\Mon,\preceq_d\}$ forms a Poset.
\end{Lemma}

The proof of this lemma is trivial. 

\begin{figure}[!h]
    \centering
    \includegraphics[width=\textwidth]{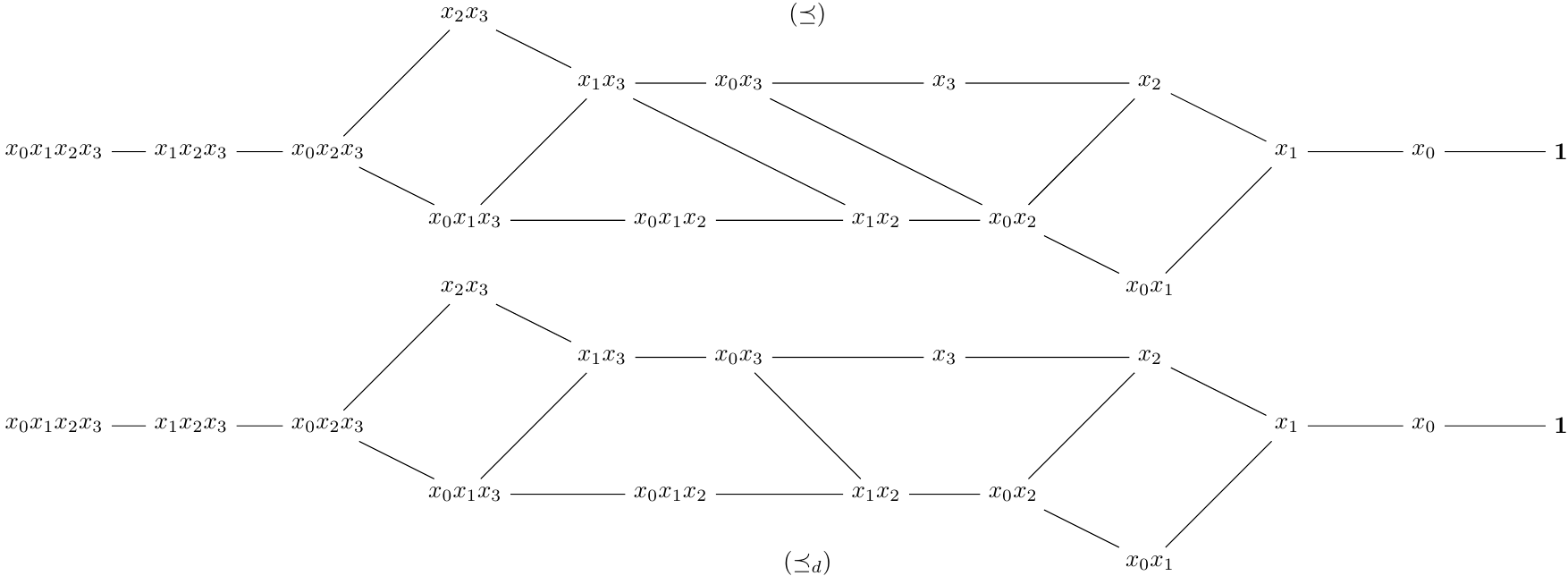}
    \caption{The two order relations $\preceq$ and $\preceq_d$ for $m=4.$}
    \label{fig:Poset}
\end{figure}

\begin{Remark}
Notice that $x_{i_1}\dots x_{i_s}  \preceq x_{j_1}\dots x_{j_s}$ implies that $x_{i_1}\dots x_{i_s}  \preceq_d x_{j_1}\dots x_{j_s}.$ The converse is no longer true, take for example the monomials $x_0x_3$ and $x_1x_2.$ 
\end{Remark}
\begin{Proposition}\label{lem:ordre_domin}
Let $f$ and $g$ be two monomial having the same degree and $x_h$ be such that $x_h\not| f$ and $x_h\not| g.$ Then we have \begin{equation}
    f\preceq_d g\text{ iff }x_hf\preceq_d x_hg.
\end{equation}  
\end{Proposition}

\begin{proof} Let $f=x_{i_1}\dots x_{i_s}$ and $g=x_{j_1}\dots x_{j_s}$ with $f\preceq_d g.$ Also, let $f^*=fx_h$ and $g^*=gx_h.$ 
There are several cases to be examined here:
\begin{itemize}
    \item If $x_{j_s}\preceq x_h$ or $x_h\preceq x_{i_1}$ then the relation $f^*\preceq_d g^*$ can easily be verified by using the definition of $\preceq_d.$
    \item If there is an integer $r\in\{1,\dots,s\}$ s.t. $x_{i_r}\preceq x_h \preceq x_{i_{r+1}}$ and $x_{j_r}\preceq x_h \preceq x_{j_{r+1}}$ then the relation $f^*\preceq_d g^*$ can easily be verified as in the previous step.
    \item If there are two distinct integers $r, t\in\{1,\dots,s\}$ s.t. $x_{i_r}\preceq x_h \preceq x_{i_{r+1}}$ and $x_{j_t}\preceq x_h \preceq x_{j_{t+1}}$ then two cases have to be considered.
    \begin{itemize}
        \item If $t<r$ then we have $x_{i_h}\preceq x_{j_{t+1}}\preceq \dots \preceq x_{j_r}$ and $x_{i_{t+1}}\preceq \dots \preceq x_{i_r}\preceq x_{h}$, which implies the following ordering \begin{equation}\label{eq:lemma-preceqD}
        x_{i+{t+1}}\preceq \dots \preceq x_{i_r}\preceq x_{h}\preceq  x_{j_{t+1}}\preceq \dots \preceq x_{j_r}.
        \end{equation}
        Combining equation \eqref{eq:lemma-preceqD} with the definition of $\preceq_d$ we obtain the desired result, i.e., $f^*\preceq_d g^*.$
        \item $t>r$ then we have $x_{i_h}\preceq x_{i_{r+1}}\preceq \dots \preceq x_{i_t}$ and $x_{j_{r+1}}\preceq \dots \preceq x_{j_t}\preceq x_{h}$ which implies the following ordering \begin{equation}\label{eq:lemma-preceqD2}
        x_{j+{r+1}}\preceq \dots \preceq x_{j_t}\preceq x_{h}\preceq  x_{i_{r+1}}\preceq \dots \preceq x_{i_t}.
        \end{equation} Now, since $x_{h}\preceq x_{i_t}$ it might be possible to have $j_s+\dots j_{t+1}+h<i_s+\dots+i_{t+1}+i_{t}$, which implies a violation of the partial sum conditions in the definition of $\preceq_d.$ If the next partial sum changes the sign, i.e., $j_s+\dots j_{t+1}+h+j_t\ge i_s+\dots+i_{t+1}+i_{t}+i_{t-1}$, by setting $\Delta_{t+1}=(j_s-i_s)+\dots+(j_{t+1}-i_{t+1})$ we have the following inequalities
        \begin{align}
            \Delta_{t+1}+h-i_t&<0\\
            \Delta_{t+1}+h-i_t+j_t-i_{t-1}&\ge 0.
        \end{align}
        This implies 
        \begin{equation}\label{eq:imp-lem}
            \Delta_{t+1}+h< i_t\leq  \Delta_{t+1}+h+j_t-i_{t-1}.
        \end{equation}
        However, since $j_t<i_{t-1}$ the equation \eqref{eq:imp-lem} becomes impossible, which completes the proof.
    \end{itemize}
    \end{itemize}
\end{proof}

In particular, notice that if $f,g$ are co-prime, i.e., $\gcd{(f,g)}\ne 1$, then \begin{equation}
    f\preceq_d g\text{ iff }{f}/{\gcd{(f,g)}}\preceq_d {g}/{\gcd{(f,g)}}.
\end{equation}

Remark that when the condition on variable $x_h$ is not satisfied, the result does not hold, e.g., $x_2x_3\preceq_d x_1x_4$, but $x_1x_2x_3$ and $x_1x_4$ are not comparable with respect to $\preceq_d.$ Before we get to our main theorem of this section, the following lemma is required.
\begin{Lemma}\label{lem:domin-deg2}
Let $f,g\in\Mon, \deg{f}=\deg{g}=2$, such that $f\preceq_d g.$ Then $\Bha {\W{f}{m}}\le\Bha{\W{g}{m}}.$   
\end{Lemma}
\begin{proof}
Let $f=x_{i_1}x_{i_2}$ and $g=x_{j_1}x_{j_2}.$ If $\gcd(f,g)\neq1$ then $f\preceq g$, which implies $\Bha{\W{f}{m}}\le\Bha{\W{g}{m}}.$ 

Now suppose $\gcd(f,g)=1$, and let $j_1-i_1\leq j_2-i_2.$ Denote $\epsilon=j_1-i_1.$ By definition of $\preceq_d$ and $\preceq$ we have 
\begin{equation}
f=x_{i_1}x_{i_2}\preceq_d x_{i_1-1}x_{i_2+1}\preceq_d\dots\preceq_d x_{j_1}x_{i_2+\epsilon}\preceq x_{j_1}x_{j_2}=g.    
\end{equation}  
If we prove $x_{i_1}x_{i_2}\preceq_d x_{i_1-1}x_{i_2+1}\Rightarrow \Bha{\W{x_{i_1}x_{i_2}}{m}}\le\Bha{\W{x_{i_1-1}x_{i_2+1}}{m}}$ the proof is finished. By definition, one can easily notice that $\Bha{\W{x_1x_2}{4}}\le\Bha{\W{x_{0}x_{3}}{4}}\Leftrightarrow\Bha{\W{x_{i_1}x_{i_2}}{m}}\le\Bha{\W{x_{i_1-1}x_{i_2+1}}{m}}.$ Hence we are left to prove that $\Bha{\W{x_1x_2}{4}}\le\Bha{\W{x_{0}x_{3}}{4}}.$ 
We have that $\Bha{\W{x_{0}x_{3}}{4}}(p)=1 - (1 - (1 - (1 - p)^2)^4)^2$ and 
   $\Bha{\W{x_1x_2}{4}}(p)=(1 - ( 1-p^2)^4)^2.$

By writing the two polynomials in the Bernstein basis, and using Theorem \ref{thm:reliab_comp} we have $\Bha{\W{x_1x_2}{4}}= \Rel(\Css^{(0,1,1,0)})$ and $\Bha{\W{x_{0}x_{3}}{4}}=\Rel(\Css^{(1,0,0,1)})$ with
\begin{align*}N_i(\Css^{(0,1,1,0)}) & = 0, 0, 0, 0, 16, 192, 1008, 3040, 5828, 7456, 6552, 4048, 1788,  560, 120, 16, 1\\
N_i(\Css^{(1,0,0,1)}) & =  0, 0, 0, 0, 32, 320, 1456, 3984, 7042, 8400, 7000, 4176, 1804, 560, 120, 16, 1
\end{align*}
As $N_i(\Css^{(0,1,1,0)})\leq N_i(\Css^{(1,0,0,1)})$ for all $i\in\{0,\dots,16\}$ we conclude the proof.
\end{proof}

\begin{Corollary}\label{cor:multipl-2}
Let $f=x_{i_1}x_{i_2}$ and $g=x_{j_1}x{j_2}$ s.t. $f\preceq_d g.$ Then for any monomial $h=x_{l_1}\dots x_{l_t}$ satisfying  $i_1<l_1<\dots<l_{t}<i_2$ we have $fh\preceq_d gh$ and $\Bha{\W{fh}{m}}\leq \Bha{\W{gh}{m}}.$ 
\end{Corollary}

\begin{Theorem}\label{thm:polar-strong-dec}
Polar codes over the Binary Erasure Channel are strongly decreasing monomial codes.
\end{Theorem}


\begin{proof}
The proof is based on two induction steps. First the parameter $m$ is fixed and we prove that the result holds for any $1\leq s\leq m$. Secondly we use induction on $m$. 

Firstly, fix $m$ and use an induction argument on the degree of monomial, namely on $s.$ We also suppose that $\gcd{(f,g)}=1.$ 
For $s=1$ we have that $\preceq=\preceq_d$ so the result is obvious. For $s=2$ use Lemma \ref{lem:domin-deg2}.

Now suppose that for any $f\preceq_d g$ with $\deg{f}=\deg{g}=s-1$ we have that $\Bha{\W{f}{m}}\le \Bha{\W{g}{m}}.$ Let $f=x_{i_1}\dots x_{i_s}$ and $g=x_{j_1}\dots x_{j_s}$ such that $g\preceq_d f$ with the usual convention $i_1 < \dots < i_s$ and $j_1 < \dots < j_s$. Then we have two cases. Either if $f/x_{i_s}\preceq_d g/x_{j_s}$ or if $x_{i_1}\leq x_{j_1}$ then we have $\Bha{\W{f}{m}}\le \Bha{\W{g}{m}}$. Indeed, in the first case we have that 
\begin{align*}
\Bha{\W{f}{m}}&=\Bha{\Wg{\Wb{x_{i_s}}{m-j_{s-1}-1}}{f/x_{i_s}}{j_{s-1}+1}}\le \Bha{\Wg{\Wb{x_{j_s}}{m-j_{s-1}-1}}{f/x_{i_s}}{j_{s-1}+1}}\\
			&\le \Bha{\Wg{\Wb{x_{j_s}}{m-j_{s-1}-1}}{g/x_{j_s}}{j_{s-1}+1}}=\Bha{\W{g}{m}}.
\end{align*} 
In the second case when $x_{i_1}\leq x_{j_1}$ the proof works in the same way. If we are not in the previous case it means that $j_1<i_1$ and $i_s<j_s.$ We know that there is $l\in\{1,\dots,s-1\}$ for which $i_l>j_l.$  First we treat the two extreme cases $l=1$ or $l=s-1.$ If $l=1$ this implies that $j_k\ge i_k$ for all $k> 1$. Let $\delta=\min\{j_2-i_2,i_1-j_1\}.$ Then
    \begin{equation*}
        f=x_{i_1}x_{i_2}\dots x_{i_s}\preceq_d x_{i_1-\delta}x_{i_2+\delta}x_{i_3}\dots x_{i_s}\preceq_d x_{j_1}\dots x_{j_s}=g.
    \end{equation*}
        
    In this case, either $x_{i_1-\delta}=x_{j_1}$ or $x_{i_2+\delta}=x_{j_2}.$ Hence, by Lemma \ref{lem:domin-deg2} and using the order relation $\preceq$ we obtain 
    
    \begin{align*}
\Bha{\W{f}{m}}&=\Bha{\Wg{\Wb{f/(x_{i_1}x_{i_2})}{m-j_{2}-1}}{x_{i_2}x_{i_1}}{j_{2}+1}}\le \Bha{\Wg{\Wb{f/(x_{i_1}x_{i_2})}{m-j_{2}-1}}{x_{i_2+\delta}x_{i_1-\delta}}{j_{2}+1}}\\
			&\le \Bha{\Wg{\Wb{x_{i_s}\dots x_{i_3}}{m-j_{2}-1}}{x_{j_2}x_{j_1}}{j_{2}+1}}\le \Bha{\Wg{\Wb{x_{j_s}\dots x_{j_3}}{m-j_{2}-1}}{x_{j_2}x_{j_1}}{j_{2}+1}}=\Bha{\W{g}{m}}.
\end{align*} 
If $l=s-1$, by putting $\delta=i_{s-1}-j_{s-1}$ and taking into account that $\delta\leq j_s-i_s$ we obtain
    \begin{align*}
\Bha{\W{f}{m}}&=\Bha{\Wg{\Wb{x_{i_s}x_{i_{s-1}})}{m-j_{s-2}-1}}{f/(x_{i_s}x_{i_{s-1}})}{j_{s-2}+1}}\leq\Bha{\Wg{\Wb{x_{i_s+\delta}x_{i_{s-1}-\delta})}{m-j_{s-2}-1}}{f/(x_{i_s}x_{i_{s-1}})}{j_{s-2}+1}}	\\
&=\Bha{\Wg{\Wb{x_{i_s+\delta}x_{j_{s-1}})}{m-j_{s-2}-1}}{x_{i_{s-2}}\dots x_{i_1}}{j_{s-2}+1}}		\leq\Bha{\Wg{\Wb{x_{j_s}x_{j_{s-1}})}{m-j_{s-2}-1}}{x_{j_{s-2}}\dots x_{j_1}}{j_{s-2}+1}}=\Bha{\W{g}{m}}.
\end{align*} 

 When $1<l<s-1$ suppose that $i_s-j_l<j_{l+1}-i_s$ and denote by $\delta_{l,s}=i_s-j_l.$  Using the definition of $\preceq_d$ we have 
 \[h=x_{j_1}\dots x_{j_l+\delta_{l,s}}x_{j_{l+1}-\delta_{l,s}}\dots x_{j_s}\preceq_d\dots\preceq_d x_{j_1}\dots x_{j_l+1}x_{j_{l+1}-1}\dots x_{j_s}\preceq_d x_{j_1}\dots x_{j_s}=g.\] 

Notice that $h=x_{j_1}\dots x_{j_{l-1}}x_{i_s}x_{j_{l+1}-i_s+j_l}\dots x_{j_s}.$ Next, we prove that $f\preceq_d h.$ Since $\gcd{(f,h)}=x_{i_s}$, we can use Lemma \ref{lem:ordre_domin}  and demonstrate ${f}/{x_{i_s}}\preceq_d {h}/{x_{i_s}}$, i.e.,\\ $x_{i_1}\dots x_{i_{s-1}}\preceq_d x_{j_1}\dots x_{j_{l-1}}x_{j_{l+1}-i_{s}+j_{l}}x_{j_{l+2}}\dots x_{j_s}.$ As,
\begin{equation}
    x_{i_{l+1}}\preceq\dots \preceq x_{i_s} \preceq x_{j_{l+1}}\preceq \dots\preceq x_{j_s},
\end{equation}
we obtain $x_{i_{l+1}}\dots x_{i_{s-1}}\preceq_d x_{j_{l+2}}\dots x_{j_s}$, simply by verifying

\begin{equation}
 \forall \;t\in\{0,\dots,s-l-2\}\quad   \sum\limits_{k=0}^{t} {i_{s-1-k}}\leq \sum\limits_{k=0}^{t} {j_{s-k}}.
\end{equation}
The next partial sums inequalities, 
\begin{equation}(i_{k}+\dots+i_{l-1})+i_{l}+i_{l+1}+\dots+i_{s-1}<(j_{k}+\dots+j_{l-1})+ j_{l+1}-i_s+j_l+j_{l+2}+\dots+j_{s}
\end{equation} are verified from the relation $f\preceq_d g.$ 
So we check the partial sums step by step:
\begin{enumerate}
\item $i_{s-1}<i_s<j_{l+1}<j_s$ and thus $x_{i_{s-1}}\preceq_d x_{j_s}$
\item $i_{s-1}+i_{s-1}<j_{s-1}+j_{s}$ and thus $x_{i_s-2}x_{i_{s-1}}\preceq_d x_{j_{s-1}}x_{j_s}$ 
\item $\dots$
\item $i_{l+1}+\dots+i_{s-1}<j_{l+2}+\dots+j_{s}$ and thus $x_{i_{l+1}}\dots x_{i_{s-1}}\preceq_d x_{j_{l+2}}\dots x_{j_s}$
\item $i_{l}+i_{l+1}+\dots+i_{s-1}<j_{l+1}-i_s+j_l+j_{l+2}+\dots+j_{s}$ by definition of $f\preceq_d g$ and thus $x_{i_{l}}\dots x_{i_{s-1}}\preceq_d x_{j_{l+1}-i_s+j_l}\dots x_{j_s}$
\item $\dots$
\end{enumerate}

Hence we have that ${f}/{x_{i_s}}\preceq_d {h}/{x_{i_s}}$ which implies, using the induction hypothesis, that 
 $   \Bha{\W{f/x_{i_s}}{m}}\le\Bha{\W{h/x_{i_s}}{m}},$
from which we deduce 
 $   \Bha{\W{f}{m}}\leq \Bha{\W{h}{m}}.$
Also,
\begin{align*}
    \Bha{\W{h}{m}}&=\Bha{\left(\Wb{x_{j_s}\dots x_{j_{l+2}}}{}^{x_{j_{l+1-\delta_{l,s}}}x_{j_{l+\delta_{l,s}}}}\right)^{x_{j_{l-1}}\dots x_{j_1}}}\\
    &=\Bha{\left((W^*)^{x_{j_{l+1-\delta_{l,s}}}x_{j_{l+\delta_{l,s}}}}\right)^{x_{j_{l-1}}\dots x_{j_1}}}\leq \Bha{\left((W^*)^{x_{j_{l+1}}x_{j_{l}}}\right)^{x_{j_{l-1}}\dots x_{j_1}}}\\
    &=\Bha{\left(\Wb{x_{j_s}\dots x_{j_{l+2}}}{}^{x_{j_{l+1}}x_{j_{l}}}\right)^{x_{j_{l-1}}\dots x_{j_1}}}=\Bha{\W{g}{m}}.
\end{align*}
Secondly, we use induction on the number of variables $m.$ For the first values of $m$, i.e., $m\leq 4$ it is straightforward to check the result.

Let $f=x_{i_1}\dots x_{i_s}$ and $g=x_{j_1}\dots x_{j_s}$ such that $g\preceq_d f$ with the usual convention $i_1 < \dots < i_s$ and $j_1 < \dots < j_s$. The following cases are possible 
\begin{itemize}
\item If $i_s=j_s=m$ then we have 
$\W{f}{m+1}  =  \Wg{\Wb{f_{m}}{1}}{f_{[0, m-1]}}{m}$ and 
$\W{g}{m+1}  =  \Wg{\Wb{g_{m}}{1}}{g_{[0,m-1]}}{m}.$
Since $f_{[0,m-1]}\preceq_d g_{[0,m-1]}$ we have by the induction hypothesis
\begin{equation}
\Bha{\Wg{\Wb{x_{m}}{1}}{f_{[0,m-1]}}{m} }\le\Bha{  \Wg{\Wb{x_{m}}{1}}{g_{[0,m-1]}}{m}}.
\end{equation}
\item Else, by the definition of the order we necessary have $j_s>i_s.$ We also have that $$h=x_{j_1}\dots x_{j_{s-1}+1}x_{j_s-1}\preceq_d x_{j_1}\dots x_{j_{s-1}}x_{j_s}=g.$$ Which implies that $\Bha{\W{h}{m+1}}\le \Bha{\W{g}{m+1}}.$ In the same time notice that $f\preceq_d h$ and $\indices(f),\indices(h)\in\{0,\dots,m-1\}.$ Therefore we obtain \begin{equation}
    \Bha{\W{f}{m+1}}\le \Bha{\W{h}{m+1}}\le \Bha{\W{g}{m+1}}.
\end{equation}
\end{itemize}
\end{proof}

\section{Average Reliability of the synthetic channels}
 
 The geometric approach of the properties of a function by means of its subgraph and/or epigraph generated useful mathematical tools from the very beginning of the theory of functions. Measure, intersection, support and shape properties lead to applications in various domains: optimization, shape description and recognition, etc. Here we propose a geometric approach in the field of polar coding. Recently, the concept of \emph{average reliability} was introduced and analyzed in the context of all terminal reliability \cite{Brown14}.   

As the synthetic channels cannot be totally ordered \cite{BDOT16,WFS17}, we propose a different method to define the optimality of a synthetic channel. For that we will check how reliable is a channel in average, i.e., we define
\begin{Definition}
Let $m$ be a strictly positive integer and $\uv\in \{0,1\}^m$. The average reliability of $W^{\uv}$ is 
\[\Avr\left(W^{\uv}\right)=\int\limits_{0}^{1}\Bha{W^{\uv}}(p)dp.\] 
Moreover, we define the relation $\LAR$ 
$$\uv\LAR\vv\Leftrightarrow \Avr({W^{\uv}})\leq\Avr({W^{\vv}})$$
\end{Definition}  

This notion of optimality has a meaning in the following context. Imagine that the communication channel is a $\BEC$ with variable erasure probability, coming from different physical reasons. This means that either we choose a different polar code in function of the variations of $p$ and in this case we obtain the best performance for each instance, or we choose a polar code and hope that in average it performs in an optimal way. The former strategy comes with the cost of computing for each value of $p$ the corresponding polar code, as for the later the cost is minimal since we only construct a polar code.    

\subsection{Properties}
\begin{Lemma}
The relation $\LAR$ is reflexive and transitive. In other words, $\LAR$ is a pre-order relation.  
\end{Lemma}

Our simulations have shown that up to $m=13$, $\LAR$ is also antisymmetric. However, this property might not be true in general. Indeed, one can easily find two distinct polynomials with integer coefficients defined over $[0,1]$ with values in $[0,1]$, such that their integrals are equal. Nonetheless, we can overcome this by applying the following procedure.

\begin{Remark}
Let $\uv\equiv_{\mathrm{Avr}}\vv$ if and only if $\Avr(\Bha{W^{\uv}})=\Avr(\Bha{W^{\vv}}).$ Let us extend the relation $\LAR$ to the factor set $\Mon/\equiv_{\mathrm{Avr}}$ naturally, using the relation between class representatives. Then $\LAR$ is a total order relation over $\Mon/\equiv_{\mathrm{Avr}}.$ Indeed, one can easily check that $\LAR$ is antisymmetric over $\Mon/\equiv_{\mathrm{Avr}}.$ 
\end{Remark}

\begin{figure}[!h]
    \begin{subfigure}{.32\textwidth}
 \centering
  \includegraphics[width=\textwidth]{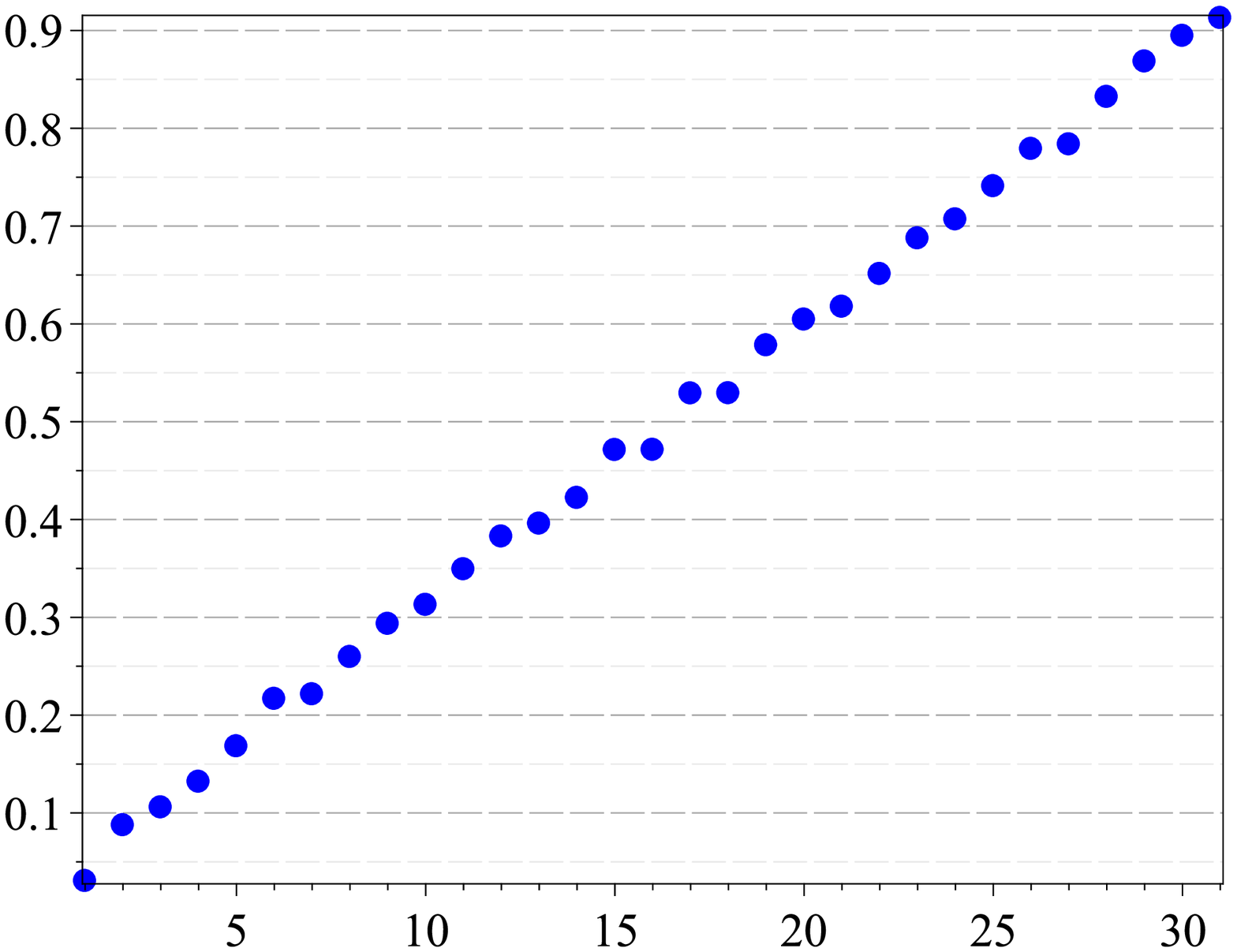}  
  \caption{$m=5$}
  \label{fig:avr5-sort}
\end{subfigure}
    \begin{subfigure}{.32\textwidth}
\centering
    \includegraphics[width=\textwidth]{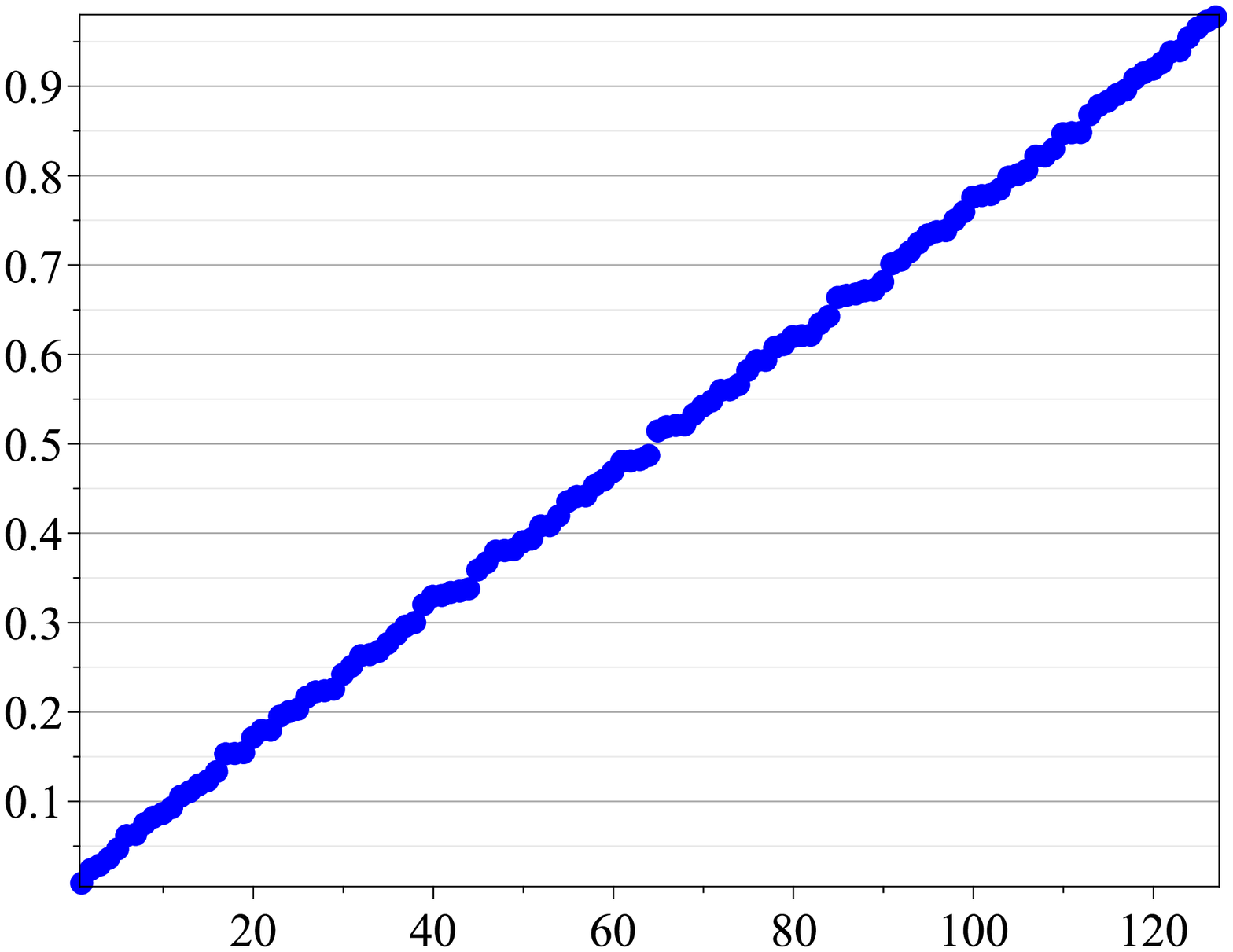}
    \caption{$m=7$}
    \label{fig:avr7-sort}
    \end{subfigure}
\begin{subfigure}{.32\textwidth}
    \centering
    \includegraphics[width=\textwidth]{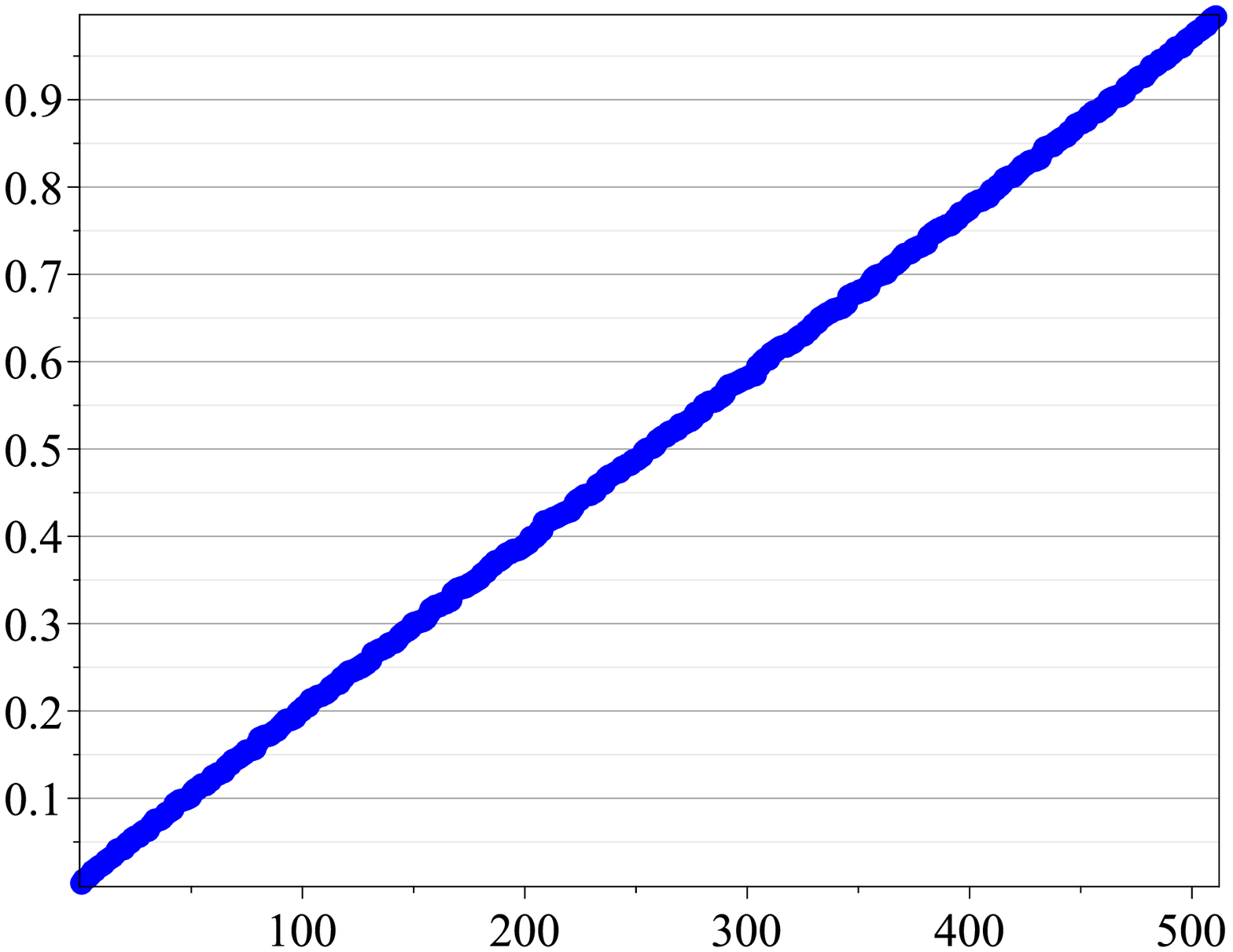}
    \caption{$m=9$}
    \label{fig:avr9-sort}
    \end{subfigure}
    \caption{Sorted $\Avr(\Bha{W^{\uv}})$ for all $\uv\in\{0,1\}^m.$}\label{fig:Avr-sort}
\end{figure}

\begin{Lemma}Let $m$ be a strictly positive integer, $n=2^m$, and $\uv\in \{0,1\}^m.$ Then
\begin{equation}
\Avr\left(W^{\uv}\right)=\dfrac{1}{n+1}\sum\limits_{i=2^{m-|\uv|}}^n\dfrac{N_i(\Css^{\uv})}{\binom{n}{i}}.
\end{equation}
\begin{equation}
\Avr\left(W^{\uv}\right)+\Avr\left(W^{\overline{\uv}}\right)=1.
\end{equation}
\end{Lemma}

\begin{Proposition}\label{pr:ord_pos_avr}
Let $m$ be a strictly positive integer and $\uv,\vv$ be two binary vectors of length $m$ such that $\uv\leq \vv$. Then 
\begin{equation} \uv\leq \vv\Rightarrow \Avr\left(W^{\uv}\right)\leq \Avr\left(W^{\vv}\right).
\end{equation}
\end{Proposition}

\begin{table}
\caption{Average reliability of the synthetic channels}\label{tab:avr}
\vspace{10pt}
\centering
  \resizebox{\textwidth}{!}{%
\begin{tabular}{|c|c|c|c|c|c|c|c|c|c|c|c|c|c|c|c|}
\hline
\multicolumn{16}{|c|}{$m=2$}\\
\hline
0&1&2&3&&&&&&&&&&&&\\
\hline
$0.20$&$0.47$&$0.53$&$0.80$&&&&&&&&&&&&\\
\hline
\hline
\multicolumn{16}{|c|}{$m=3$}\\
\hline
0&1&2&4&3&5&6&7&&&&&&&&\\
\hline
$0.11$&$0.29$&$0.34$&$0.41$&$0.59$&$0.66$&$0.71$&$0.89$&&&&&&&&\\
\hline
\hline
\multicolumn{16}{|c|}{$m=4$}\\
\hline
0&1&2&4&8&3&5&6&9&10&12&7&11&13&14&15\\
\hline
$0.06$&$0.16$&$0.20$&$0.24$&$0.30$&$0.38$&$0.44$&$0.48$&$0.52$&$0.56$&$0.62$&$0.70$&$0.76$&$0.80$&$0.84$&$0.94$\\
\hline
\end{tabular}
}
\end{table}

In Table \ref{tab:avr} we compute the $\Avr\left(W^{\uv}\right)$ all the binary vectors $\uv\in \{0,1\}^m$ for $m \in \{2,3,4\}.$ Notice that in this case Proposition \ref{pr:ord_pos_avr} applies, since we know that up to $m=4$ the synthetic channels can be totally ordered over the $\BEC$ \cite{D17,WFS17}. Starting from $m=5$ this property is no longer true. When $\uv$ and $\vv$ are no longer comparable, i.e., there is $p_0\in (0,1)$ such that $\Bha{W^{\uv}}(p_0)=\Bha{W^{\vv}}(p_0)$, we can still decide whether in average $\uv$ is optimal compared with $\vv.$ The set of non comparable pairs $(\uv,\vv)$ for $m=5$ is $\{(3,16),(12,17),(7,20),(7,24),(11,24),(14,19),(15,28)\}.$ Notice that half of the pairs are coming from duality, i.e., if $(\uv,\vv)$ are not comparable then $(\overline{\uv},\overline\vv)$ are also non-comparable. However, these are ordered with respect to average reliability.  The average reliability for the first 4 non-comparable pairs are $(0.221,0.216), (0.396,0.383),$ $(0.4712,0.4710), (0.4712, 0.5288)$. Hence, for $m=5$ the ordering with respect to the average reliability is $0,1,2,4,8,16,3,5,6,9,10,17,12,18,20,7,24$, and the rest can be completed by symmetry. 

\begin{figure}[!h]
 \hrulefill
\begin{Example}The ordering induced by the average reliability.
\begin{itemize}
\item $m=5$
\[0, \underbrace{1, 2, 4, 8, 16}_{\RM(1,5)},\overbrace{ 3, 5, 6, 9, 10, 17, 12, 18, 20}^{\RM(2,5)},\underbrace{7}_{\RM(3,5)} ,\overbrace{ 24}^{\RM(2,5)}\]
\item[\noindent]
\item $m=6$
\[0, \underbrace{1, 2, 4, 8, 16}_{\RM(1,6)}, \overbrace{3, 5}^{\RM(2,6)},\underbrace{ 32}_{\RM(1,6)}, \overbrace{6, 9, 10, 17, 12, 18, 33, 20}^{\RM(2,6)},\underbrace{ 7}_{\RM(3,6)}, \overbrace{34, 24}^{\RM(2,6)},\underbrace{ 11}_{\RM(3,6)},\overbrace{ 36}^{\RM(2,6)},\]\[\underbrace{  13, 19, 14}_{\RM(3,6)} , \overbrace{ 40}^{\RM(2,6)},\underbrace{  21}_{\RM(3,6)}, \overbrace{ 48}^{\RM(2,6)}, \underbrace{ 22, 35, 25, 37, 26, 38, 28, 41}_{\RM(3,6)}\]
\end{itemize}
\end{Example}

 \hrulefill
\vspace*{4pt}
\end{figure}

Our simulations have shown that, considering the relation $\LAR$  in the set of the synthetic channels, in each sub-interval $(i/10,(i+1)/10])$, for $0\leq i\leq 9$, we have a rough proportion of $2^m/10$ binary vectors $\uv.$ So, roughly speaking an uniform distribution could be used to approximate the number of $\uv$ inside each sub-interval (illustrated in Figure \ref{fig:Avr-sort}), with respect to $\Avr$ (see Table \ref{tab:distrib-avr} for $5\leq m\leq 11.$) 
\begin{table}[!h]
    \centering
    \caption{Number of $\uv\in\{0,1\}^m$ that satisfy $\Avr(\Bha{W^{\uv}})\in (i/10,(i+1)/10]$, for $0\leq i<5$,  $\epsilon_m=2^{m-4}/10.$}
    \label{tab:distrib-avr}
    \begin{tabular}{|c||c|c|c|c|c||c|}
         \hline
         $m$&$(0,0.1]$&$(0.1,0.2]$&$(0.2,0.3]$&$(0.3,0.4]$&$(0.4,0.5]$&$[\lfloor 2^m/10-\epsilon_m\rfloor,\lceil 2^m/10+\epsilon_m\rceil]$ \\
         \hline\hline
        5 &2&3&4&4&3&$[3,4]$\\
        \hline
         6&5&7&6&8&6&$[6,7]$\\
        \hline
     7&11&13&14&13&13&$[12,14]$\\
        \hline
        8&23&25&27&27&26&$[24,28]$\\
        \hline
        9&49&51&50&55&51&$[48,55]$\\
        \hline
        10&99&104&98&107&104&$[97,109]$\\
        \hline
        11&199&209&204&204&208&$[194,218]$\\
        \hline
    \end{tabular}
\end{table}

\subsection{Relation to $\beta$-expansion}

$\beta$-expansion \cite{He_2017} is a well-known method for an efficient construction of Polar codes. Hence, it is with no surprise that our results on average reliability determine possible more refined choices of the variable $\beta.$ Let us begin by defining the method.
\begin{equation}
    \beta(\uv)=\sum\limits_{i=0}^{m-1}u_i\beta^i
\end{equation}

In \cite{He_2017} the authors proved that for any $\beta\in(1,\infty)$ the order induced by $\beta$ on the sequence of synthetic channels respects the order relation $\preceq.$ In particular this means that if $\uv\preceq \vv$ then $\beta(\uv)\leq\beta(\vv)$ and this for any value of $\beta>1.$ Some values of $\beta$ are of high interest, in particular $\beta=2^{1/4}$, when $W$ is designed for Additive White Gaussian Noise (AWGN). In the case of AWGN, the authors in \cite{He_2017} proposed a procedure in which an interval for $\beta$ is determined, interval that converges to a value close to $2^{1/4}.$ Notice that in \cite{He_2017} the order induced by $\beta$ is not valid for any Signal to Noise Ratio value, but it tries to cover as much as possible the interval $[0,1].$ A natural question that one could raise is whether there is a $\beta$-expansion for the average reliability, i.e., is there a real value $\beta$ such that $\beta$ and $\Avr$ are identical over the set of binary vectors of length $m.$ There is a significant difference between the two relations. In our case, not only that $W$ is a BEC but also the pre-order induced by the average reliability is total over $\Mon/\equiv_{\mathrm{Avr}}$ and holds for the entire interval $[0,1].$  

\begin{Remark}
By computer simulations one can easily prove that for $m\leq 5,$ there is $\beta\in(1,\infty)$ such that the order induced by $beta$ and the pre-order induced by the average reliability coincide. It can be done by simply tacking $\beta=1.22$.
\end{Remark}

\begin{Conjecture}
For $m>6$ we did not find a value of $\beta$ for which the two aforementioned relations are equal.
Moreover, for $\beta\sim 1.22$ the number of elements having similar mutual relations with respect to the two relations is minimized (see Table \ref{tab:non-comp-beta-avr}).   
\end{Conjecture}

\begin{table}[!h]
    \centering
        \caption{Number of pairs $(\uv,\vv)$ satisfying $\Avr(\Bha{W^{\uv}})\leq \Avr(\Bha{W^{\vv}})$ for which $\not\exists \beta\in(1,\infty)$ s.t. $\beta(\uv)\leq \beta(\vv).$  }
    \label{tab:non-comp-beta-avr}
    \begin{tabular}{c|c|c|c}
         \hline
         $m$&$\beta$&number of incompatible pair of elements & $2^m$\\
        \hline\hline
         4& $(1,1.32]$& --& 16\\
         5& $(1.18,1.22]$& --& 32\\
         6& $1.22$&2&64\\
         7&$1.22$ & 10&128\\
         8&$1.22$ & 36 &256\\
         9&$1.22$ & 99 &512\\
         \hline
    \end{tabular}
\end{table}

\subsection{Threshold points of the binary erasure polarization sub-channels}

The fact that when $m$ goes to infinity the Bhattacharyya polynomial has a sharp transition from zero to one when $m$ goes to infinity has already been proven (\cite{mondelli2016}).
\begin{Lemma}
\begin{equation}
\lim\limits_{m\to \infty}\Bha{W^{\uv}}=\left\{
\begin{array}{cc}
0 &p\in [0,p_0(\uv))\\
1&p \in (p_0(\uv),1]\\
\end{array}
\right.
\end{equation}
\end{Lemma}

However, finding the point $p_0(\uv)$ where this transition holds is not trivial (see \cite{Roth2019,He_2017}). Here we will use the average reliability to determine this point for some specific channels.  
\begin{Lemma}\label{pr:change_lim_avr}
\begin{equation}
\lim\limits_{m\to \infty}\Avr\left(W^{\uv}\right)=1-p_0(\uv).
\end{equation}
\end{Lemma}
A particular interesting channel analyzed in \cite{Roth2019,WFS17} is the synthetic channel $W^{(1^i0^{m-i})}.$ More exactly, the authors analyze the sharp transition of  $W^{(1^i0^{m-i})}$ from $0$ to $1$ when $m$ tends to infinity, in function of the limit ${i}/{m-i}.$ Here, we will give an exact formula for the average reliability of $W^{(1^i0^{m-i})}$. This result combined with Lemma \ref{pr:change_lim_avr} will allow us to obtain a finer approximation of $p_0(\uv).$ To achieve our goal, we will look at the corresponding 2TN, namely at $\Css^{(1^i0^{m-i})}$. For simplification we use $l=2^{m-i}, w=2^{i}$ and $n=2^m.$ Notice that
\begin{equation} 
\Rel\left(\Css^{(1^i0^{m-i})};p\right)=1-(1-p^l)^w.
\end{equation}

\begin{Theorem}\label{thm:rel_POS} \begin{equation}
\Rel\left(\Css^{(1^i0^{m-i})};p\right)=\sum\limits_{i=l}^{n}\sum\limits_{j=1}^{\lfloor\frac{i}{l}\rfloor}(-1)^{j+1}\binom{w}{j}\binom{n-jl}{n-i}p^i(1-p)^{n-i}.
\end{equation}
\end{Theorem}
\begin{proof}
In order to prove our result we need to demonstrate that $\forall \;l\leq i\leq n$
\begin{equation}
N_i(\Css^{(1^i0^{m-i})})=\sum\limits_{j=1}^{\lfloor\frac{i}{l}\rfloor}(-1)^{j+1}\binom{w}{j}\binom{n-jl}{n-i}
\end{equation}

The proof is based on an inclusion-exclusion argument. Denote by $\mathcal{P}_i$ the set of paths of length $i$ from $S$ to $T$ for the $\Css^{(1^i0^{m-i})}$. This leads to $\left|\mathcal{P}_i\right|=N_i(\Css^{(1^i0^{m-i})}).$ 

Any path of length $i$ with $l\leq i$ is composed of at least one path of length $l$, hence we have $w$ choices for fixing a path of length $l$ and $\binom{n-l}{i}$ choices for the remaining positions. However, in the  $\binom{n-l}{i}$ choices we might count other $l$ length paths. Hence, we need to subtract the over-counting, which is all the combinations of two length $l$ paths, i.e., $\binom{w}{2}$, times the number of choices for the remaining positions, i.e., $\binom{n-2l}{i-2l}.$ Now, we need to add all the paths that are composed of at least $3$ $l$ paths which equals $\binom{w}{3}\binom{n-3l}{i-3l}.$ And so on till we reached the last level, i.e., $\binom{w}{\lfloor\frac{i}{l}\rfloor}\binom{n-l\lfloor\frac{i}{l}\rfloor}{i-l\lfloor\frac{i}{l}\rfloor}.$

\end{proof}

\begin{Theorem}\label{thm:avr_pos}
\begin{equation}
\Avr\left(W^{(1^i0^{m-i})}\right)=1-\dfrac{1}{\binom{2^i+2^{i-m}}{2^i}}
\end{equation}
\end{Theorem}

\begin{proof}
\begin{align*}
\Avr\left(W^{\uv}\right)&=\dfrac{1}{n+1}\sum\limits_{i=l}^n\dfrac{N_i(\Css^{\uv})}{\binom{n}{i}}=\dfrac{1}{n+1}\sum\limits_{i=l}^{n}\sum\limits_{j=1}^{\lfloor\frac{i}{l}\rfloor}(-1)^{j+1}\dfrac{\binom{w}{j}\binom{n-jl}{n-i}}{\binom{n}{i}}\\
&=\dfrac{1}{n+1}\sum\limits_{j=1}^w\sum\limits_{i=jl}^{n}(-1)^{j+1}\dfrac{\binom{w}{j}\binom{n-jl}{n-i}}{\binom{n}{i}}=\dfrac{1}{n+1}\sum\limits_{j=1}^w\sum\limits_{i=jl}^{n}(-1)^{j+1}\dfrac{\binom{w}{j}\binom{i}{jl}}{\binom{n}{jl}}\\
&=\dfrac{1}{n+1}\sum\limits_{j=1}^w(-1)^{j+1}\dfrac{\binom{w}{j}}{\binom{n}{jl}}\sum\limits_{i=jl}^{n}\binom{i}{jl}=\dfrac{1}{n+1}\sum\limits_{j=1}^w(-1)^{j+1}\dfrac{\binom{w}{j}\binom{n+1}{jl+1}}{\binom{n}{jl}}\\
&=\sum\limits_{j=1}^w(-1)^{j+1}\binom{w}{j}\dfrac{1}{jl+1}=1-\sum\limits_{j=0}^w(-1)^{j}\binom{w}{j}\dfrac{1}{jl+1}=1-\dfrac{1}{\binom{\frac{n+1}{l}}{w}}
\end{align*}
\end{proof}

Strightforward, we have
\begin{Corollary}
\begin{equation}
\Avr\left({W^{(0^i1^{m-i})}}\right)=\dfrac{1}{\binom{2^i+2^{i-m}}{2^i}}
\end{equation}
\end{Corollary}

Based on Theorem \ref{thm:avr_pos} we can establish new classes of asymptotically "good" channels. For that we will need the following result.
\begin{Lemma}\label{lem:asympt-binom}
\begin{align}
    \lim\limits_{n\to\infty}\begin{pmatrix}
    \frac{n}{\log_2(n)(\log_2(\log_2(n)))}+\frac{1}{\log_2(n)(\log_2(\log_2(n)))}\\ \frac{n}{\log_2(n)(\log_2(\log_2(n)))}\end{pmatrix} &=1.\\
    \lim\limits_{n\to\infty}\begin{pmatrix}
    \frac{n}{\log_2(\log_2(n))}+\frac{1}{\log_2(\log_2(n))}\\ \frac{n}{\log_2(\log_2(n))}\end{pmatrix} &=\infty.\\
        \lim\limits_{n\to\infty}\begin{pmatrix}
    \frac{n}{\log_2(n)}+\frac{1}{\log_2(n)}\\ \frac{n}{\log_2(n)}\end{pmatrix} &=2.
\end{align}    
\end{Lemma}

Theorem \ref{thm:avr_pos}, Lemma \ref{lem:asympt-binom} and Lemma \ref{pr:change_lim_avr} imply the following result.
    
\begin{Corollary}
Let $m$ be a strictly positive integer and $\uv=\left(1^i0^{m-i}\right)$ Then
\begin{itemize}
\item for any $i\leq m-\log_2(m)-\log_2(\log_2(m))$ we have $p_0(\uv)\to 1$ and $p_0(\overline{\uv})\to 0.$
\item for any $i\ge m-\log_2(\log_2(m))$ we have $p_0(\uv)\to 0$ and $p_0(\overline{\uv})\to 1.$
\end{itemize}
\end{Corollary}

Another direct consequence of our results is that for any $i\leq m-\log_2(m)-\log_2(\log_2(m))$ the monomial $f=x_0\dots x_{i-1}$ is highly reliable in average. Hence, all the monomials $g\preceq_d f$ are also highly reliable in average, as their average reliability tends to zero when $m$ goes to infinity. Also, $f$ becomes unreliable in average for $i\geq m-\log_2(\log_2(m)).$ The values $m-\log_2(m)-\log_2(\log_2(m))< i < m-\log_2(\log_2(m))$ are to be considered in more details.  

\begin{Corollary}
Let $m$ be a strictly positive integer and $i\leq \log_2(\log_2(m)).$ Then for any $f\in\Mon$ with $f\preceq_d x_{m-i+1}\dots x_m$ we have that $\Avr(W^{f})\to 0$ when $m\to\infty.$ In other words, any synthetic channel in the $\RM(i,m)$ is asymptotically "good" in average.  
\end{Corollary}
\section{Conclusions and Perspectives}
A complete characterization of the Bhattacharyya parameter of synthetic channels of a monomial code is an open problem that has attracted a lot of attention in the last decade. Even for the particular case of Binary Erasure Channel the question remains unanswered. However, the implications of such a result are of high importance in coding theory, specially in polar coding. In this article, we make a step forward by proposing an order relation $\preceq_d$ that decreases the gap between state-of-the-art and the ultimate partial order relation for the Battacharya parameter of synthetic channels. The advantage of this approach is that our algebraic description is rather easy to implement and analyze, compared to other order relations such as \cite{WFS17}. Simulations show that $\preceq_d$ is a valid order relation on Binary Symmetric Channel, and a deeper inspection of \cite{WFS17} and our work, could potentially determine an algebraic description that fits the latest results.            

As the relations on the Bhattacharyya parameter are all partial orders, we have proposed an alternative solution for ordering the synthetic channels. For that, we have used the concept of average reliability, borrowed from the network theory. Instead of the local evaluation of the Bhattacharyya parameter, we propose a global one, by evaluating the integral, i.e., by measuring its global average behavior. Hence, we rank the synthetic channels using a pre-order relation $\LAR$, given by the value of the integral. Our result is not constructive, in the sense that it does not fully characterize the channels that belong to a specific interval. An answer to this question might provide an extremely efficient method for constructing polar codes and give much more insight on the synthetic channels $W^{\uv}.$ 

\section*{Acknowledgements} V-F. Dragoi is supported by a grant of the Romanian Ministry of Education and Research, CNCS- UEFISCDI, project number PN-III-P1-1.1-PD-2019-0285, within PNCDI III.

\bibliographystyle{abbrv}

\end{document}